\documentclass[11pt,british,refpage,intoc,bibliography=totoc,index=totoc,BCOR=7.5mm,captions=tableheading]{extarticle}
\usepackage{mathptmx}
\usepackage[T1]{fontenc}
\usepackage[latin9]{inputenc}
\usepackage[a4paper]{geometry}
\usepackage{color}
\usepackage{babel}
\usepackage{amsmath}
\usepackage{amsthm}
\usepackage{amssymb}
\usepackage{graphicx}
\usepackage[numbers]{natbib}
\usepackage[unicode=true,pdfusetitle,
 bookmarks=true,bookmarksnumbered=true,bookmarksopen=false,
 breaklinks=false,pdfborder={0 0 1},backref=false,colorlinks=true]
 {hyperref}
\hypersetup{
 pdfborderstyle=,linkcolor=black,citecolor=blue,urlcolor=black,filecolor=blue,pdfpagelayout=OneColumn,pdfnewwindow=true,pdfstartview=XYZ,plainpages=false}

\makeatletter

\providecommand{\tabularnewline}{\\}

\theoremstyle{plain}
\newtheorem{thm}{\protect\theoremname}
\theoremstyle{plain}
\newtheorem{lem}[thm]{\protect\lemmaname}

\@ifundefined{date}{}{\date{}}
\usepackage{babel}
\usepackage{caption}
\usepackage[nottoc]{tocbibind}
\allowdisplaybreaks[4]
\usepackage{dsfont}
\DeclareMathAlphabet{\mathcal}{OMS}{cmsy}{m}{n}

\providecommand{\lemmaname}{Lemma}
\providecommand{\theoremname}{Theorem}

\makeatother

\providecommand{\lemmaname}{Lemma}
\providecommand{\theoremname}{Theorem}

\begin{document}
\global\long\def\Sgm{\boldsymbol{\Sigma}}%
\global\long\def\W{\boldsymbol{W}}%
\global\long\def\H{\boldsymbol{H}}%
\global\long\def\P{\mathbb{P}}%
\global\long\def\Q{\mathbb{Q}}%

\title{Monte-Carlo simulations for wall-bounded fluid flows\\
 via random vortex method}
\author{By Z. Qian\thanks{Mathematical Institute, University of Oxford, Oxford OX2 6GG, UK.,
and OSCAR, Suzhou, China. Email: \protect\protect\href{mailto:qianz@maths.ox.ac.uk}{qianz@maths.ox.ac.uk}} , Y. Qiu\thanks{Institut de Math\'ematiques de Toulouse, UMR 5219, Universit\'e
de Toulouse, CNRS, UPS, F-31062, Toulouse Cedex 9, France.}, L. Zhao\thanks{Oxford Suzhou Centre for Advanced Research, Suzhou, China},
and J. Wu\thanks{Department of Mathematics, Swansea University, Swansea SA1 8EN, UK.
Email: j.l.wu@swansea.ac.uk } }
\maketitle
\begin{abstract}
In this paper a Monte-Carlo method for simulating the motion of fluid
flow moving along a solid wall is proposed. The random vortex method
in the present paper is established by using the reflection technology
and perturbation technique. The Monte-Carlo method based on this random
vortex dynamic may be implemented, and several Monte-Carlo simulations
are then carried out for the flows near the solid wall.

\medskip{}

\emph{Key words}: diffusion processes, incompressible fluid flow,
Monte-Carlo simulation, random vortex method

\medskip{}

\emph{MSC classifications}: 76M35, 76M23, 60H30, 65C05, 68Q10, 
\end{abstract}

\section{Introduction}

Our goal in this work is to present a new numerical method for computing
solutions of the Navier-Stokes equations which describe viscous fluid
flows passing a plate. The technology in the paper may be generalised
to three dimensional case, while substantial modifications are required,
hence in the present work only two dimensional case will be dealt
with and the three dimensional case will be published in separate
articles. The numerical schemes for solving fluid dynamics equations
are based on the random vortex dynamics, which were developed in the
past mainly for turbulent flows away from physical boundaries.

The key idea in the random vortex method proposed in Chorin \citep{Chorin 1973}
may be briefly described for two dimensional flows as the following,
cf. \citep{Goodman1987}, \citep{Long1988}, \citep{CottetKoumoutsakos2000}
and \citep{Majda and Bertozzi 2002} for details. A fluid flow may
be described by its velocity vector field $u(x,t)$, and the fundamental
problem in fluid dynamics (and in the study of turbulence) is to extract
information about the velocity $u(x,t)$ from its equations of motion.
It is well known that, regardless of the nature of fluid flows, determining
the vector field $u(x,t)$ is equivalent to determining its integral
curves $X^{\xi}$ of the dynamics 
\[
\frac{\textrm{d}}{\textrm{d}t}X_{t}^{\xi}=u(X_{t}^{\xi},t),\quad X_{0}^{\xi}=\xi,
\]
where $t\rightarrow X_{t}^{\xi}$ is considered as the trajectory
of the ``fluid'' particle issued from location $\xi$. In order
to employ the idea of Monte-Carlo methods, one may consider $u(x,t)$
as the velocity of ``imaginary'' Brownian fluid particles (this
idea to the best knowledge of the present authors is due to G. I.
Taylor \citep{Taylor1921}), whose trajectories, still denoted by
$X^{\xi}$ (while it is a random process on a probability space),
may be determined by solving It{\^o}'s stochastic differential equation
(SDE, cf. \citep{Ikeda and Watanabe 1989}) 
\begin{equation}
\textrm{d}X_{t}^{\xi}=u(X_{t}^{\xi},t)\textrm{d}t+\sqrt{2\nu}\textrm{d}B_{t},\quad X_{0}^{\xi}=\xi\label{sde1}
\end{equation}
for $\xi\in\mathbb{R}^{2}$, where $\nu>0$ is a constant (which will
be the kinematic viscosity of the fluid). The key idea in any Monte-Carlo
simulation is to average these trajectories of Brownian fluid particles
and recover therefore the flow velocity $u(x,t)$. In order to implement
this strategy, one first solves \emph{the closure problem:} by utilising
the motion equation for $u(x,t)$ and the distribution of $X$ to
eliminate formally the vector field $u(x,t)$ in \eqref{sde1} in
terms of the distribution of the random field $X(x,t)$, and therefore
the SDE \eqref{sde1} is closed. In random vortex methods, this is
achieved by using the motion equation of its vorticity $\omega=\nabla\wedge u$,
a vector field with its components $\omega^{i}=\varepsilon^{ijk}\frac{\partial u^{k}}{\partial x_{j}}$.
This scheme can be explained best for two dimensional incompressible
fluid flows. Let us explicate this point in more details below.

For a two dimensional fluid flow with viscosity $\nu>0$, the velocity
$u(x,t)$ satisfies the Navier-Stokes equations 
\begin{equation}
\frac{\partial u}{\partial t}+(u\cdot\nabla)u-\nu\Delta u+\nabla P=0,\quad\textrm{ and }\nabla\cdot u=0,\label{t-02-1}
\end{equation}
so that its vorticity $\omega=\nabla\wedge u$, a scalar function,
evolves according to the vorticity transport equation: 
\begin{equation}
\frac{\partial}{\partial t}\omega+(u\cdot\nabla)\omega-\nu\Delta\omega=0,\quad\textrm{ and }\nabla\wedge u=\omega.\label{V-001}
\end{equation}
Hence $\omega$ has the following integral representation
\begin{equation}
\omega(x,t)=\int p_{u}(0,\xi,t,x)\omega_{0}(\xi)\textrm{d}\xi,\label{rep-t-01}
\end{equation}
where $\omega_{0}$ is the initial vorticity, and $p_{u}(\tau,\xi,t,x)$
is the transition probability density function of Taylor's diffusion
defined by \eqref{sde1} (whose infinitesimal generator is $\nu\Delta+u\cdot\nabla$).
$p_{u}(\tau,\xi,t,x)$ coincides with the Green function associated
with the forward operator $\frac{\partial}{\partial t}-\nu\Delta+u\cdot\nabla$.
Since $\nabla\cdot u=0$ and $\nabla\wedge u=\omega$, the Biot-Savart
law holds: 
\begin{equation}
u(x,t)=-\int G(x,\cdot)\nabla\wedge\omega=\int K(x,y)\omega(y,t)\textrm{d}y,\label{Gr-001}
\end{equation}
where $G$ is the Green function of the Laplacian on $\mathbb{R}^{2}$
and $K=\nabla G$. Putting the representations for $\omega$ and for
$u$ together we obtain 
\begin{align*}
u(x,t) & =\int\int K(x,y)p_{u}(0,\xi,t,y)\omega_{0}(\xi)\textrm{d}\xi\textrm{d}y\\
 & =\int\mathbb{E}\left[K(x,X_{t}^{\xi})\right]\omega_{0}(\xi)\textrm{d}\xi.
\end{align*}
Therefore the SDE \eqref{sde1} can be closed in the following manner
\begin{equation}
\begin{cases}
\textrm{d}X_{t}^{\xi}=\left.\int\mathbb{E}\left[K(x,X_{t}^{\eta})\right]\omega_{0}(\eta)\textrm{d}\eta\right|_{x=X_{t}^{\xi}}\textrm{d}t+\sqrt{2\nu}\textrm{d}B_{t},\\
X_{0}^{\xi}=\xi\in\mathbb{R}^{2}.
\end{cases}\label{cl-2D}
\end{equation}

SDE (\ref{cl-2D}) can be used for designing numerical schemes for
computing numerically the velocity field $u(x,t)$ accordingly.

While, like many other beautiful ideas in mathematics, when applying
these ideas to concrete situations, one has to overcome several obstacles.
This is not exceptional when one applies the ideas of random vortex
methods to wall-bounded flows, and we need to reformulate the random
vortex method which suits for the study of wall-bounded flows. One
has to consider not only the boundary condition for the velocity,
but also the possible constraint on the vorticity motion at the solid
wall. The vorticity at the wall can not be determined before one is
able to solve the Navier-Stokes equations, which therefore poses substantial
difficulty in implementing the random vortex method. We make the following
key observation: the Biot-Savart law holds well for a bounded domain
as the velocity needs to satisfy the no slip condition, while we may
recover the velocity $u(x,t)$ not only in terms of its vorticity
$\omega$, but we may also obtain $u(x,t)$ by using some perturbations
of $\omega$ which satisfy the right boundary condition but still
solve certain parabolic equations. Therefore we are able to formulate
a family of random vortex dynamics for wall-bounded flows by using
the reflection principle and one parameter family of perturbations
of the vorticity.

We should mention that computational fluid dynamics (CFD) is a huge
subject, and various numerical approaches have been developed such
as Direct Numerical Simulations (DNS) (cf. \citep{OrszagPatterson1972},
\citep{Spalart1988}, \citep{MoinMashesh1998}, see also \citep{Cuvelier-Segal-van Steenhoven1986,Temam2000,Wesseling2001,Sengupta-Bhaumik2019}
and the references therein for details), Large Eddy Simulations (LES)
(cf. \citep{Smagorinsky1963}, \citep{Lilly1967}, \citep{BerselliIliescuLayton2006,LesieurMetaisComte2005}
and the literature cited therein for further reading), Probability
Density Function (PDF) (cf. \citep{Pope2000} and \citep{LaskariMcKeon}
for example) and other technologies have been developed, and have
become increasingly important. A few numerical experiments for fluid
flows, including turbulent flows, within their boundary layers, have
been carried out, which provide substantial information about wall-bounded
fluid flows, for example \citep{ChauhanPhilipetl2014}, \citep{HEISEL2018},
\citep{MCKEON2010}, \citep{DawsonMcKeon2019}, \citep{HeadBandyopadhyay1981},
\citep{RaiMoin1993}, \citep{SpalartWatmuff1993}, \citep{WuMoin2008,WuMoin2009,WuMoinHickey2014}
and etc. Of course there is a large number of papers dealing with
many different aspects of the wall-bounded flows, the authors must
apologise for failing to mention many excellent contributions, and
the interested reader may consult \citep{HirschelCousteixKordulla2014},
\citep{Keller1978}, \citep{Schlichting9th-2017} and the papers cited
therein. Stochastic simulation schemes have been studied too, mainly
under the name of random vortex methods which were initiated in \citep{Chorin 1973},
cf. \citep{CottetKoumoutsakos2000} and \citep{Majda and Bertozzi 2002}
for a comprehensive account. While to the best knowledge of the present
authors, random vortex methods are mainly developed for fluid flows
moving freely without boundary constraint. The goal of the present
paper is therefore to formulate a Monte-Carlo method for numerically
calculating solutions for simple wall bounded incompressible fluid
flows. According to an important discovery by Prandtl \citep{Prandtl1904},
the viscosity however how small demonstrates its effect in a thin
layer close to the solid wall, in order to ensure the fluid flow to
obey the no slip boundary condition. The equations of motions in a
thin boundary-layer can be simplified to the boundary layer equation
which is much simpler than the Navier-Stokes equations. The integration
of the boundary layer equation may be used to determine some very
important aspects of wall bounded fluid flows such as the stress immediately
applied to the solid wall, which provides with us the necessary information,
together with the stochastic representations we are going to establish,
to implement Monte-Carlo simulations for boundary turbulent layer
flows. This line of research will be addressed however in a future
work.

The rest of the paper is organised as follows. Section 2 presents
several technical lemmas for later use. In Section 3, two dimensional
flows are investigated, and a Monte-Carlo method is established in
Section 4. The paper ends with several simulation experiments for
the Monte-Carlo simulation.

\section{The vorticity transport equation}

In this paper we aim to implement the random vortex method and propose
a Monte-Carlo scheme for numerically simulating the fluid motion near
the thin boundary layer. We consider a fluid flow moving in the half
space 
\begin{equation}
D=\left\{ x=(x_{1},\cdots,x_{d-1},x_{d})\in\mathbb{R}^{d}:x_{d}<0\right\} \label{dom-01-1}
\end{equation}
where $d=2$ or $3$, whose velocity is denoted by $u=(u^{1},\ldots,u^{d})$
is a time-dependent vector field in $D$. The velocity $u$ is determined
by the Navier-Stokes equations: 
\begin{equation}
\frac{\partial}{\partial t}u+(u\cdot\nabla)u-\nu\Delta u+\nabla P=F\quad\textrm{ in }D\label{meq-01}
\end{equation}
and 
\begin{equation}
\nabla\cdot u=0\quad\textrm{ in }D,\label{meq-02}
\end{equation}
where $P$ is the pressure and $F$ is the external force applied
to the fluid. The velocity $u$ has to satisfy the no slip condition
which says that $u(x,t)$ vanishes for $x\in\partial D$ and $t>0$.
By taking exterior derivative of both sides of the Navier-Stokes equation,
we obtain 
\[
\frac{\partial}{\partial t}\omega+(u\cdot\nabla)\omega-\nu\Delta\omega-(\omega\cdot\nabla)u=G\quad\textrm{ in }D
\]
where $G=\nabla\wedge F$. The boundary value of $\omega$ along the
solid wall $\partial D$ (it will be called the boundary vorticity
in short) can not be specified, and must be computed via (\ref{meq-01},
\ref{meq-02}).

Our first task is to identify the boundary vorticity in terms of measurable
fluid dynamics variables. Observe that $\frac{\partial}{\partial x_{d}}$
is the unit normal pointing outward and $\frac{\partial}{\partial x_{1}}$,
$\ldots$, $\frac{\partial}{\partial x_{d-1}}$ form a basis of the
tangent space along $\partial D$. It is perhaps convenient to introduce
the following notations. If $T$ is a tensor field then $T^{\perp}$
and $T^{\parallel}$ denote the normal and tangential parts of $T$
respectively. Let $S=(S_{ij})$ be the symmetric tensor field, called
the rate-of-strain. By definition 
\begin{equation}
S_{ij}=\frac{\partial u^{i}}{\partial x_{j}}+\frac{\partial u^{j}}{\partial x_{i}}\label{stress-01}
\end{equation}
so along its boundary, the normal components of the stress tensor
are $S_{13}$, $S_{23}$ and $S_{33}$. For an incompressible fluid,
$\textrm{tr}(S_{ij})=0$, together with the no slip condition, it
follows that $S_{11}=S_{22}=S_{33}=0$ on $\partial D$. Therefore
the normal part of the stress $S^{\perp}=(S_{13},S_{23})$, which
are the shear stress.

If $d=3$, the vorticity $\omega=\nabla\wedge u$ has its components
$\omega^{i}=\varepsilon^{ijk}\frac{\partial u^{j}}{\partial x_{k}}$,
so that the boundary vorticity 
\begin{equation}
\left.\omega^{1}\right|_{\partial D}=-\left.\frac{\partial u^{2}}{\partial x_{3}}\right|_{\partial D},\quad\left.\omega^{2}\right|_{\partial D}=\left.\frac{\partial u^{1}}{\partial x_{3}}\right|_{\partial D},\quad\left.\omega^{3}\right|_{\partial D}=0.\label{bd-c1}
\end{equation}
Due to the no slip boundary condition, 
\[
\left.\frac{\partial u^{3}}{\partial x_{2}}\right|_{\partial D}=\left.\frac{\partial u^{3}}{\partial x_{1}}\right|_{\partial D}=0
\]
therefore 
\begin{equation}
\left.\omega^{1}\right|_{\partial D}=-\left.S_{23}\right|_{\partial D},\quad\left.\omega^{2}\right|_{\partial D}=\left.S_{13}\right|_{\partial D},\quad\left.\omega^{3}\right|_{\partial D}=0.\label{bd-c2}
\end{equation}
Hence the tangent components $\omega^{\Vert}$ of the vorticity $\omega$
along the solid wall can be identified with the normal components
$S^{\perp}$ of the stress at the wall.

For a two dimensional flow, where $d=2$, the vorticity $\omega$
can be identified with the scalar function $\omega=\frac{\partial u^{2}}{\partial x_{1}}-\frac{\partial u^{1}}{\partial x_{2}}$.
As we have indicated already, the vorticity transport equation has
a simpler form: 
\begin{equation}
\frac{\partial}{\partial t}\omega+(u\cdot\nabla)\omega-\nu\Delta\omega=G,\quad\textrm{ in }D,\label{2d-v4}
\end{equation}
that is, the vorticity equation is a scalar ``linear'' parabolic
equation (if $u$ is considered as a known dynamic variable) without
non-linear stretching term. It is however in contrast with the case
of isotropic turbulence, turbulence may be built up for two dimensional
flows (i.e. flows with certain symmetries) due to boundary layer phenomena
near the solid wall. For two dimensional flows, the boundary vorticity
\begin{equation}
\left.\omega\right|_{\partial D}=-\left.S_{12}\right|_{\partial D}\label{2d-shear-s1}
\end{equation}
where $S_{12}$ is the (normal) rate-of-strain at the solid wall.

\section{Random vortex dynamics}

In this section the mathematical framework is set up for implementing
Monte-Carlo simulations for wall-bounded flows. The method is based
on the random vortex dynamics, and the key step is to derive a functional
integral representations for a family of modified vorticity dynamical
variables.

Some of the techniques work for 
\begin{equation}
D=\left\{ x=(x_{1},\cdots,x_{d-1},x_{d})\in\mathbb{R}^{d}:x_{d}<0\right\} .\label{dom-01}
\end{equation}
for any dimension $d\geq2$, while the reflection in $\mathbb{R}^{d}$
about the hyperspace: $x_{d}=0$ is the mapping which sends $x=(x_{1},\cdots,x_{d-1},x_{d})$
to $\bar{x}=(x_{1},\cdots,x_{d-1},-x^{d})$.

\subsection{The Biot-Savart law}

The first ingredient needed for deriving the random vortex dynamics
is a version of the Biot-Savart law for $D$. By definition $\omega=\nabla\wedge u$
and $\nabla\cdot u=0$, hence $\Delta u=-\nabla\wedge\omega$ in $D$.
Since $u$ vanishes on $\partial D$, according to Green formula 
\begin{equation}
u(x,t)=-\int_{D}\varGamma_{D}(x,y)\nabla\wedge\omega(y,t)\textrm{d}y\label{green-01}
\end{equation}
for $x\in D$, where $\varGamma_{D}$ denotes the Green kernel of
$D$. Suppose $u(x,t)$ and $\omega(x,t)$ decay to zero sufficiently
fast, then we may perform integration by parts, using the no slip
condition for $u$, it follows that 
\begin{equation}
u(x,t)=\int_{D}K(x,y)\wedge\omega(y,t)\textrm{d}y\label{3d-bs-law}
\end{equation}
where $K(x,y)=\nabla_{y}\varGamma_{D}(x,y)$, which is called the
Biot-Savart law for $D$. In dimension two one may prefer to formulate
it in slightly different form as $\omega$ is scalar. That is, in
dimension two 
\begin{equation}
u^{i}(x,t)=\int_{D}K^{i}(x,y)\omega(y,t)\textrm{d}y,\quad\textrm{ for }i=1,2,\label{2d-bs-law}
\end{equation}
where 
\begin{equation}
K^{1}(x,y)=\frac{\partial}{\partial y_{2}}\varGamma_{D}(x,y),\quad K^{2}(x,y)=-\frac{\partial}{\partial y_{1}}\varGamma_{D}(x,y).\label{K12-2d}
\end{equation}

It will reduce the computational cost if an explicit formula for the
singular integral kernel $K(x,y)$ is available. To this end we recall
that the Green function in $\mathbb{R}^{d}$ is given by 
\begin{equation}
\varGamma(x,y)=\frac{1}{2\pi}\log|x-y|\quad\textrm{ for }d=2\label{Green-01}
\end{equation}
and 
\begin{equation}
\varGamma(x,y)=\frac{1}{\kappa_{d}(2-d)}|x-y|^{2-d}\quad\textrm{ if }d\geq3\label{Green-02}
\end{equation}
where $\kappa_{d}$ is the area of the unit sphere in $\mathbb{R}^{d}$.
By the reflection principle, the Green function for $D$ (subject
to the Dirichlet boundary condition) has an explicit formula 
\begin{equation}
\varGamma_{D}(x,y)=\varGamma(x,y)-\varGamma(x,\bar{y})\label{Green-03}
\end{equation}
for $x,y\in D$.

For example in dimension two we have the following facts which can
be established for $D$ (which is unbounded) with slight modifications
of the arguments in any standard textbooks. 
\begin{lem}
\label{lem3.2}Suppose $u\in C^{2}(D)\cap C^{1}(\overline{D})$ solves
the Poisson equation 
\begin{equation}
\Delta u=f\quad\textrm{ in }D,\quad\left.u\right|_{\partial D}=0\label{Pss-01}
\end{equation}
where $f\in L^{1}(D)$, and $u\rightarrow0$ and $\nabla u\rightarrow0$
at infinity. Then 
\begin{equation}
u(x)=\int_{D}\varGamma_{D}(x,y)f(y)\textrm{d}y\quad\textrm{ for }x\in D.\label{rep-031}
\end{equation}
\end{lem}

\begin{proof}
The only fact we need is the asymptotic of $\varGamma_{D}(x,y)$ as
$y\rightarrow\infty$. If $x=(x_{1},x_{2})\in D$ and $y=(y_{1},y_{2})\in D$,
then 
\begin{equation}
\varGamma_{D}(x,y)=\frac{1}{4\pi}\log\left(1-\frac{4x_{2}y_{2}}{|y-\bar{x}|^{2}}\right)\quad\textrm{ for }x,y\in D,\label{2D-G-01}
\end{equation}
which implies that 
\begin{equation}
\varGamma_{D}(x,y)\sim-\frac{1}{\pi}\frac{x_{2}y_{2}}{|y-\bar{x}|^{2}}\quad\textrm{ as }|y|\rightarrow\infty.\label{2D-G-02}
\end{equation}
The Green formula then follows from the standard argument. 
\end{proof}
\begin{lem}
\label{lem3.3}Suppose $u=(u^{1},u^{2})\in C^{2}(D)\cap C^{1}(\overline{D})$,
$u=0$ on $\partial D$, such that $u\rightarrow0$ and $|\nabla u|\rightarrow0$
at the infinity and $\nabla\cdot u=0$ in $D$. Let $\omega=\nabla\wedge u$
and assume that both $\omega,\nabla\omega\in L^{1}(D)$. Then 
\begin{equation}
u^{i}(x)=\int_{D}K^{i}(x,y)\omega(y)\textrm{d}y\quad\textrm{ for all }x\in D\label{B-S law}
\end{equation}
for $i=1,2$, where 
\begin{equation}
K^{1}(x,y)=\frac{1}{2\pi}\left(\frac{y_{2}-x_{2}}{|y-x|^{2}}-\frac{y_{2}+x_{2}}{|y-\bar{x}|^{2}}\right)\label{eq:qq12-1}
\end{equation}
and 
\begin{equation}
K^{2}(x,y)=\frac{1}{2\pi}\left(\frac{y_{1}-x_{1}}{|y-\bar{x}|^{2}}-\frac{y_{1}-x_{1}}{|y-x|^{2}}\right).\label{eq:qq13-1}
\end{equation}
\end{lem}

\begin{proof}
This follows from the Green formula and integration by parts. 
\end{proof}
The following lemma provides an elementary fact which will be used
later on. 
\begin{lem}
\label{lem2.5-1}Let $a<b$ be two numbers and $x=(x_{1},x_{2})\in D$
fixed. Then $K^{1}(x,y)$ and $K^{2}(x,y)$ are integrable on $[a,b]\times(-\infty,0)$.
Let 
\[
H^{i}(x)=\int_{(a,b)\times(-\infty,0)}K^{i}(x,y)\textrm{d}y
\]
where $i=1,2$. Then 
\begin{equation}
H^{1}(x)=-\frac{1}{2}x_{2}\left(\textrm{sgn}(b-x_{1})-\textrm{sgn}(a-x_{1})\right)\label{K1-int01-1}
\end{equation}
and 
\begin{equation}
H^{2}(x)=\frac{x_{2}}{2\pi}\ln\frac{(b-x_{1})^{2}+x_{2}^{2}}{(a-x_{1})^{2}+x_{2}^{2}}+\frac{1}{\pi}\left((b-x_{1})\arctan\frac{x_{2}}{b-x_{1}}-(a-x_{1})\arctan\frac{x_{2}}{a-x_{1}}\right).\label{K2-int-02}
\end{equation}
In particular 
\begin{equation}
\int_{-\infty}^{0}\int_{x_{1}-A}^{x_{1}+A}K^{2}(x,y)\textrm{d}y_{1}\textrm{d}y_{2}=0\label{K2-int03}
\end{equation}
for every $A>0$ and $x=(x_{1},x_{2})$ with $x_{2}<0$. 
\end{lem}

\subsection{Taylor's diffusion}

Suppose $b(x,t)=(b^{1}(x,t),\cdots,b^{d}(x,t))$ is a time dependent
vector field on $\overline{D}$, which is differentiable up to the
boundary $\partial D$, bounded, Borel measurable, and vanishes along
the boundary $\partial D$. The vector field $b(x,t)$ is extended
to the whole space $\mathbb{R}^{d}$ through reflection. That is,
if $x_{d}>0$, then 
\begin{equation}
b^{i}(x,t)=b^{i}(\bar{x},t)\quad\textrm{ for }i=1,\ldots,d-1\textrm{ and }b^{d}(x,t)=-b^{d}(\bar{x},t).\label{v1-01}
\end{equation}
Since $b(x,t)=0$ along $\partial D$, $b(x,t)=\overline{b(t,\bar{x}})$
for all $x\in\mathbb{R}^{d}$.

We will assume that $\nabla\cdot b(x,t)=0$ for $x_{d}>0$. Then the
extension via the reflection is divergence-free, that is, $\nabla\cdot b(\cdot,t)=0$
in $\mathbb{R}^{d}$ in distribution sense. For such a time dependent
vector field we consider the differential operator of second order
on $\mathbb{R}^{d}$:

\begin{equation}
L_{b(x,t)}=\nu\Delta+b(x,t)\cdot\nabla\label{eq:form 2}
\end{equation}
where the differential operators $\Delta$ and $\nabla$ apply only
to the space variable $x$. The variables $(x,t)$ in the sub-script
will be omitted if no confusion may arise. Since \textbf{$\nabla\cdot b=0$}
in distribution, the formal adjoint operator $L_{b}^{\star}=L_{-b}$,
which is again a diffusion operator of the same type.

Suppose $\Omega\subset\mathbb{R}^{d}$, then $\Gamma_{\Omega,b}(x,t;\xi,\tau)$;
where $0\leq\tau<t$ and $\xi,x\in\Omega$; denotes the \emph{Green
function} to the (forward) parabolic equation 
\begin{equation}
\left(\frac{\partial}{\partial t}-L_{b(x,t)}\right)f(x,t)=0\quad\textrm{ in }(0,\infty)\times\Omega,\label{eq:f-par-aa1}
\end{equation}
subject to the Dirichlet boundary condition that $\left.f(\cdot,t)\right|_{\partial\Omega}=0$,
in the sense that (i) for every $\xi\in\Omega$, $\tau\geq0$, as
a function $(x,t)$, $f(x,t)=\Gamma_{\Omega,b}(x,t;\xi,\tau)$ solves
the previous boundary problem of (\ref{eq:f-par-aa1}) for $t>\tau$;
and for every bounded and continuous function $\varphi$ on $\Omega$
\begin{equation}
\lim_{t\downarrow\tau}\int_{\Omega}\varphi(\xi)\Gamma_{\Omega,b}(x,t;\xi,\tau)\textrm{d}\xi=\varphi(x)\label{eq:int-fund-01}
\end{equation}
for all $x\in\Omega$.

If $\Omega=\mathbb{R}^{d}$, then we will use $\Gamma_{b}$ to denote
$\Gamma_{\mathbb{R}^{d},b}$ for simplicity.

Similarly, $\Gamma_{\Omega,b}^{\star}(x,t;\xi,\tau)$ (defined for
$0\leq t<\tau$, $x,\xi\in\mathbb{R}^{n}$) denotes a Green function
to the backward parabolic equation 
\begin{equation}
\left(\frac{\partial}{\partial t}+L_{b(x,t)}^{\star}\right)f(x,t)=0\quad\textrm{ in }\Omega.\label{eq:back-par-01}
\end{equation}
Since $L_{b}^{\star}=L_{-b}$, so that the backward equation can be
written as 
\[
\left(\frac{\partial}{\partial t}+L_{-b(x,t)}\right)f(x,t)=0\quad\textrm{ in }\Omega.
\]

\begin{lem}
\label{lem3.1} Suppose that $\nabla\cdot b=0$ in distribution and
$b(x,t)=\overline{b(\bar{x},t)}$ for all $x\in\mathbb{R}^{d}$, $t>0$.
Then 
\begin{equation}
\Gamma_{D,b}(x,t;\xi,\tau)=\frac{1}{2}\left(\Gamma_{b}(x,t;\xi,\tau)-\Gamma_{b}(x,t;\bar{\xi},\tau)-\Gamma_{b}(\bar{x},t;\xi,\tau)+\Gamma_{b}(\bar{x},t;\bar{\xi},\tau)\right)\label{f-01}
\end{equation}
and 
\begin{equation}
\Gamma_{D,b}^{\star}(x,t;\xi,\tau)=\frac{1}{2}\left(\Gamma_{b}^{\star}(x,t;\xi,\tau)-\Gamma_{b}^{\star}(x,t;\bar{\xi},\tau)-\Gamma_{b}^{\star}(\bar{x},t;\xi,\tau)+\Gamma_{b}^{\star}(\bar{x},t;\bar{\xi},\tau)\right)\label{f-02}
\end{equation}
for any $x,\xi\in D$ and $\tau<t$. 
\end{lem}

\begin{proof}
Since $b(x,t)=\overline{b(\bar{x},t)}$, so by definition, it is easy
to see that $\Gamma_{b}(x,t;\xi,\tau$), $\Gamma_{b}(x,t;\bar{\xi},\tau)$,
$\Gamma_{b}(\bar{x},t;\xi,\tau)$ and $\Gamma_{b}(\bar{x},t;\bar{\xi},\tau)$
are solutions to the parabolic equation 
\[
\left(\frac{\partial}{\partial t}-L_{b(x,t)}\right)f(x,t)=0
\]
in $\mathbb{R}^{d}$ for $t>\tau$. It follows that 
\[
\Gamma(x,t;\xi,\tau)=\frac{1}{2}\left(\Gamma_{b}(x,t;\xi,\tau)-\Gamma_{b}(x,t;\bar{\xi},\tau)-\Gamma_{b}(\bar{x},t;\xi,\tau)+\Gamma_{b}(\bar{x},t;\bar{\xi},\tau)\right)
\]
solves the boundary problem and $\left.f(\cdot,t)\right|_{x=\bar{x}}=0$,
and 
\[
\lim_{t\downarrow\tau}\int_{D}\varphi(\xi)\Gamma(x,t;\xi,\tau)\textrm{d}\xi=\varphi(x).
\]
Therefore $\Gamma=\Gamma_{D,b}$. 
\end{proof}
It is known that $\Gamma_{b}(x,t;\xi,\tau)=\Gamma_{b}^{\star}(\xi,\tau;x,t)$
where $t>\tau$ and $x,\xi\in\mathbb{R}^{d}$, so that as a consequence,
we have $\Gamma_{D,b}(x,t;\xi,\tau)=\Gamma_{D,b}^{\star}(\xi,\tau;x,t)$.

Let $p_{b}(s,x,t,y)$ denote the transition probability density function
of the diffusion process $X$ with its infinitesimal generator $L_{b}$.
$X$ may be constructed as a (weak) solution to the stochastic differential
equation: 
\[
\textrm{d}X=b(X,t)\textrm{d}t+\sqrt{2\nu}\textrm{d}B,\quad X_{\tau}=\xi
\]
where $B$ is a Brownian motion on some probability space. Then 
\[
p_{b}(\tau,\xi,t,x)\textrm{d}x=\mathbb{P}\left[\left.X_{t}=\textrm{d}x\right|X_{\tau}=\xi\right]
\]
for $\tau<t$. Since $L_{b}^{\star}=L_{-b}$ as $b$ is divergence-free,
so that 
\begin{equation}
p_{b}(\tau,\xi,t,x)=\Gamma_{-b}^{\star}(\xi,\tau;x,t)=\Gamma_{-b}(x,t;\xi,\tau),\label{def-r1}
\end{equation}
by Lemma \ref{lem3.1}, 
\begin{equation}
\Gamma_{D,b}(x,t;\xi,\tau)=p_{D,-b}(\tau,\xi,t,x)\label{eq:convers-y1}
\end{equation}
for $t\geq\tau$ and $x,\xi\in D$, where $p_{D,-b}(\tau,\xi,t,x)$
is the transition probability density function of the $L_{-b}$-diffusion
killed on leaving $D$. On the other hand, since $b(x,t)=0$ for $x\in\partial D$
and $\overline{b(\bar{x},t)}=b(x,t)$, so we must have 
\begin{equation}
p_{b}(\tau,\xi,t,x)=p_{b}(\tau,\bar{\xi},t,\bar{x})\quad\textrm{ for any }\xi,x\in\mathbb{R}^{d}\textrm{ and }t>\tau.\label{f-03}
\end{equation}
By combining \eqref{f-02}, \eqref{def-r1} and \eqref{f-03} together,
we deduce the following lemma. 
\begin{lem}
\label{lem2.5}Under the same assumptions on $b(x,t)$ as in Lemma
\ref{lem3.1}. It holds that 
\begin{equation}
\Gamma_{D,b}(x,t;\xi,\tau)=p_{D,-b}(\tau,\xi,t,x)=p_{-b}(\tau,\xi,t,x)-p_{-b}(\tau,\xi,t,\bar{x})\label{f-04}
\end{equation}
for $\tau<t$ and $\xi,x\in D$. 
\end{lem}

As a consequence we have the following lemma which provides another
ingredient needed in formulating the random vortex system. 
\begin{lem}
\label{lem2.6}Under the same assumptions on $b(x,t)$ as in Lemma
\ref{lem3.1}. If $w(x,t)$ is the solution to the initial value problem
\begin{equation}
\left(\frac{\partial}{\partial t}-L_{b(x,t)}\right)w(x,t)=g(x,t)\textrm{ in }D\times[\tau,\infty)\label{eq:cauchy prm1-1}
\end{equation}
for $t>\tau$, with initial data that $w(x,\tau)=\varphi(x)$, and
satisfies the Dirichlet boundary condition that $\left.w(x,t)\right|_{x\in\partial D}=0$
for $t>0$. Then 
\begin{align}
w(x,t) & =\int_{D}\left(p_{-b}(\tau,\xi,t,x)-p_{-b}(\tau,\bar{\xi},t,x)\right)\varphi(\xi)\textrm{d}\xi\nonumber \\
 & +\int_{\tau}^{t}\int_{D}\left(p_{-b}(s,\xi,t,x)-p_{-b}(s,\bar{\xi},t,x)\right)g(\xi,s)\textrm{d}\xi\textrm{d}s.\label{rep-w-01}
\end{align}
\end{lem}

\begin{proof}
By definition we have 
\begin{equation}
w(x,t)=\int_{D}\Gamma_{D,b}(x,t;\xi,\tau)\varphi(\xi)\textrm{d}\xi+\int_{\tau}^{t}\int_{D}\Gamma_{D,b}(x,t;\xi,s)g(\xi,s)\textrm{d}s\textrm{d}\xi.\label{eq:rep-fun1-1}
\end{equation}
Thanks to the duality (\ref{eq:convers-y1}) this equality may be
rewritten as 
\begin{equation}
w(x,t)=\int_{D}p_{D,-b}(\tau,\xi,t,x)\varphi(\xi)\textrm{d}\xi+\int_{\tau}^{t}\int_{D}p_{D,-b}(s,\xi,t,x)g(\xi,s)\textrm{d}\xi\textrm{d}s\label{eq:rep-tra1-1}
\end{equation}
for $t>\tau$ for $x\in D$, and the representation follows from \eqref{f-03}
and \eqref{f-04} immediately. 
\end{proof}

\subsection{\label{Sec33}Two dimensional wall-bounded flows}

From now on we only consider the two dimensional case, so that $D=\{(x_{1},x_{2}):x_{2}<0\}$
and the vorticity $\omega=\frac{\partial u^{2}}{\partial x_{1}}-\frac{\partial u^{1}}{\partial x_{2}}$
evolves according to the vorticity transport equation (\ref{2d-v4})
in $D$, which is a solution to the linear parabolic equation (\ref{2d-v4})
if $u$ is considered given. While unlike the whole space case, unfortunately,
one is unable to apply the representation (\ref{rep-w-01}) to $\omega$
as we have pointed already, $\omega$ satisfies a non-homogeneous
boundary condition containing unknown data. Therefore in order to
apply Lemma \ref{lem2.6}, it remains to deal with the boundary vorticity,
which is given in (\ref{2d-shear-s1}).

To simplify our arguments, we make two technical assumptions: 
\begin{itemize}
\item We assume that the initial velocity $u_{0}(x)=u(x,0)$ is smooth up
to the boundary, and of course we assume that $u_{0}(x_{1},0)=0$
for all $x_{1}$. We may assume that $u_{0}(x_{1},x_{2})=0$ if $x_{1}\notin(a,b)$
where $a<b$ are two numbers. 
\item The external force $F=(F^{1},F^{2})$ is smooth, has a compact support
in $\overline{D}$, and vanishes at the boundary $\partial D$. 
\end{itemize}
These assumptions can be weaken greatly, and they are imposed for
simplifying our derivation of the random vortex dynamics below.

The velocity $u(x,t)$ satisfies the no slip condition, so that it
can be extended to the whole space $\mathbb{R}^{2}$, still denoted
by $u(x,t)$, such that $u(x,t)=\overline{u(\bar{x},t)}$ for every
$x\in\mathbb{R}^{2}$ and $t\geq0$. $u(x,t)$ is divergence-free
in $\mathbb{R}^{2}$ in distribution sense. Let $p_{u}(\tau,\xi,t,x)$
be the probability transition function of Taylor's diffusion, i.e.
a diffusion with its infinitesimal generator $\nu\Delta+u\cdot\nabla$.

For simplicity denote the boundary vorticity $\left.\omega\right|_{\partial D}=-\left.\frac{\partial u^{1}}{\partial x_{2}}\right|_{\partial D}$
by $\theta$. Then $\theta$ is a function on $\partial D:x_{2}=0$,
so $\theta$ depends only on $x_{1}$. Under our assumptions 
\[
\theta(x_{1},t)=-\frac{\partial u^{1}}{\partial x_{2}}(x_{1},0,t)\quad\textrm{ for }x_{1}\in\mathbb{R}\textrm{ and }t\geq0.
\]
We introduce a family of perturbations of the vorticity $\omega$
modified near the boundary by using a cut-off function, so that the
modified vorticiy vanishes along the boundary $\partial D$. More
precisely $\theta$ is extended to the interior of $D$ as the following:
for any given $\varepsilon>0$ set 
\begin{equation}
\sigma_{\varepsilon}(x_{1},x_{2},t)=\theta(x_{1},t)\phi(-x_{2}/\varepsilon),\label{ext-01}
\end{equation}
where $\phi:[0,\infty)\rightarrow[0,1]$ is a proper cut-off function
to be chosen later, and $\phi$ is smooth, such that $\phi(r)=1$
for $r\in[0,1/3)$ and $\phi(r)=0$ for $r\geq2/3$. Let $W^{\varepsilon}=\omega-\sigma_{\varepsilon}$.
Then it is easy to verify that 
\begin{equation}
\left(\frac{\partial}{\partial t}+u\cdot\nabla-\nu\Delta\right)W_{\varepsilon}=g_{\varepsilon}\quad\textrm{ in }D,\quad\textrm{ and }\left.W_{\varepsilon}\right|_{\partial D}=0,\label{eq:qq4}
\end{equation}
where 
\begin{align}
g_{\varepsilon}(x,t) & =G(x,t)+\frac{\nu}{\varepsilon^{2}}\phi''(-x_{2}/\varepsilon)\theta(x_{1},t)+\frac{1}{\varepsilon}\phi'(-x_{2}/\varepsilon)u^{2}(x,t)\theta(x_{1},t)\nonumber \\
 & +\phi(-x_{2}/\varepsilon)\left(\nu\frac{\partial^{2}\theta}{\partial x_{1}^{2}}(x_{1},t)-\frac{\partial\theta}{\partial t}(x_{1},t)\right)-\phi(-x_{2}/\varepsilon)u^{1}(x,t)\frac{\partial\theta}{\partial x_{1}}(x_{1},t)\label{g-xt-01}
\end{align}
for any $x=(x_{1},x_{2})$, $x_{2}\geq0$. The initial data for $W^{\varepsilon}$
is identified with the following 
\[
W_{0}^{\varepsilon}(x)=\omega_{0}(x_{1},x_{2})-\omega_{0}(x_{1},0)\phi(-x_{2}/\varepsilon)\quad\textrm{ for }x\in D.
\]

Under our technical assumptions, $K(x,y)\sigma_{\varepsilon}(y,t)$
for every $x$ and $t$, as a function of $y$, is integrable on $D$.
Therefore, according to the integral representation (cf. Lemma \ref{lem2.6})
applying to \eqref{eq:qq4} we then obtain that

\begin{align}
\omega(y,t) & =\int_{D}\left(p_{u}(0,\xi,t,y)-p_{u}(0,\bar{\xi},t,y)\right)W_{0}^{\varepsilon}(\xi)\textrm{d}\xi+\sigma_{\varepsilon}(y,t)\nonumber \\
 & +\int_{0}^{t}\int_{D}\left(p_{u}(s,\xi,t,y)-p_{u}(s,\bar{\xi},t,y)\right)g_{\varepsilon}(\xi,s)\textrm{d}\xi ds\label{eq:qq6}
\end{align}
for $y\in D$ and $t>0$. On the other hand $u(x,t)$ can be recovered
via the Biot-Savart law (cf. Lemma \ref{lem3.3}) to obtain 
\begin{align}
u(x,t) & =\int_{D}\int_{D}K(x,y)\left(p_{u}(0,\xi,t,y)-p_{u}(0,\bar{\xi},t,y)\right)W_{0}^{\varepsilon}(\xi)\textrm{d}\xi\textrm{d}y+\int_{D}K(x,y)\sigma_{\varepsilon}(y,t)\textrm{d}y\nonumber \\
 & +\int_{D}\int_{0}^{t}\int_{D}K(x,y)\left(p_{u}(s,\xi,t,y)-p_{u}(s,\bar{\xi},t,y)\right)g_{\varepsilon}(\xi,s)\textrm{d}\xi\textrm{d}s\textrm{d}y\label{key-m01}
\end{align}
where $i=1,2$, $x\in D$ and $t>0$, and $K=(K^{1},K^{2})$ is given
by (\ref{eq:qq12-1}, \ref{eq:qq13-1}).

The next step is to express the integrals involving the transition
probability density $p_{u}(0,\xi,t,y)$ in terms of the distribution
of the diffusion with infinitesimal generator $\nu\Delta+u\cdot\nabla$.
Recall that, if $(X_{t}^{\xi,s})$ is Taylor's diffusion started from
$\xi\in\mathbb{R}^{2}$ at instance $s\geq0$, that is, a solution
to the stochastic differential equation 
\begin{equation}
\textrm{d}X_{t}^{\xi,s}=u(X_{t}^{\xi,s},t)\textrm{d}t+\sqrt{2\nu}\textrm{d}B_{t},\quad X_{s}^{\xi,s}=\xi\label{X-sde1}
\end{equation}
for $t\geq s$ (of course we may define $X_{t}^{\xi,s}=\xi$ for all
$t\leq s$), where $B$ is a two dimensional Brownian motion, then
\[
\int_{D}f(y)p_{u}(s,\xi,t,y)\textrm{d}y=\mathbb{E}\left[1_{D}(X_{t}^{\xi,s})f(X_{t}^{\xi,s})\right].
\]
By using this elementary fact and Fubini's theorem we may rewrite
(\ref{key-m01}) and obtain the following general representation theorem. 
\begin{thm}
\label{thm:g}Let $u(x,t)$ be a solution to Navier-Stokes equations
(\ref{meq-01}, \ref{meq-02}) in $D=\{x:x_{2}>0\}$. Let $\varepsilon>0$
and a cut-off function $\phi$ be given. Then 
\begin{align}
u(x,t) & =\int_{D}K(x,y)\sigma_{\varepsilon}(y,t)\textrm{d}y\nonumber \\
 & +\int_{D}\mathbb{E}\left[1_{D}(X_{t}^{\xi,0})K(x,X_{t}^{\xi,0})-1_{D}(X_{t}^{\bar{\xi},0})K(x,X_{t}^{\bar{\xi},0})\right]W_{0}^{\varepsilon}(\xi)\textrm{d}\xi\nonumber \\
 & +\int_{0}^{t}\int_{D}\mathbb{E}\left[1_{D}(X_{t}^{\xi,s})K(x,X_{t}^{\xi,s})-1_{D}(X_{t}^{\bar{\xi},s})K(x,X_{t}^{\bar{\xi},s})\right]g_{\varepsilon}(\xi,s)\textrm{d}\xi\textrm{d}s,\label{main formula}
\end{align}
for $x\in D$, $u(x,t)=\overline{u(\overline{x},t)}$ for $x_{2}>0$,
and $u(x,t)=0$ if $x_{2}=0$, where $X^{\xi,s}$ are defined by (\ref{X-sde1}),
$W^{\varepsilon}$ and $g_{\varepsilon}$ are defined as above. 
\end{thm}

The stochastic representation (\ref{main formula}) holds well for
any choice of $\varepsilon>0$ and $\phi$. This is an advantage for
implementing Monte-Carlo simulations. To complete the scheme we need
to deal with the extended stress $\sigma_{\varepsilon}(x,t)$ and
the term $g_{\varepsilon}(x,t)$. For implementing Monte-Carlo simulations
based on (\ref{main formula}), we should choose $\varepsilon>0$
to be small. In fact one should choose $\varepsilon>0$ to be much
smaller than the boundary layer thickness $\delta$, so that the velocity
terms appearing in $g_{\varepsilon}$ can be ignored. The boundary
stress term $\theta$ has to be computed either via dynamics method
or boundary layer equations. 
\begin{thm}
\label{thm:main-g} Let $u(x,t)$ be a solution to Navier-Stokes equations
(\ref{meq-01}, \ref{meq-02}) in $D=\{x:x_{2}<0\}$, and $G=\nabla\wedge F$.
Then $u(x,t)$ and the Taylor diffusion with its infinitesimal generator
$\nu\Delta+u\cdot\nabla$ form the closed random vortex dynamics:
\begin{align}
u(x,t) & =\int_{D}K(x,y)\theta(y_{1},t)\textrm{d}y\nonumber \\
 & +\int_{D}\mathbb{E}\left[1_{D}(X_{t}^{\xi,0})K(x,X_{t}^{\xi,0})-1_{D}(X_{t}^{\bar{\xi},0})K(x,X_{t}^{\bar{\xi},0})\right]W_{0}(\xi)\textrm{d}\xi\nonumber \\
 & +\int_{0}^{t}\int_{D}\mathbb{E}\left[1_{D}(X_{t}^{\xi,s})K(x,X_{t}^{\xi,s})-1_{D}(X_{t}^{\bar{\xi},s})K(x,X_{t}^{\bar{\xi},s})\right]g(\xi,s)\textrm{d}\xi\textrm{d}s,\textrm{ for }x\in D,\label{main formula-2}
\end{align}
\begin{equation}
u(x,t)=\overline{u(\overline{x},t)}\textrm{ for }x_{2}>0,\textrm{ and }u(x,t)=0\textrm{ if }x_{2}=0,\label{m-2-01}
\end{equation}
and 
\begin{equation}
\textrm{d}X_{t}^{\xi,s}=u(X_{t}^{\xi,s},t)\textrm{d}t+\sqrt{2\nu}\textrm{d}B_{t},\quad X_{s}^{\xi,s}=\xi\quad\textrm{ for }s\geq0\textrm{ and }\xi\in\mathbb{R}^{2},\label{sde-m01}
\end{equation}
where $\omega_{0}(x)=\nabla\wedge u_{0}(x)$, 
\begin{equation}
W_{0}(x)=\omega_{0}(x_{1},x_{2})+\frac{\partial u_{0}}{\partial x_{2}}(x_{1},0)\label{W0-01}
\end{equation}
$u_{0}(x)=u(x,0)$ are the initial data, and 
\begin{equation}
g(x,t)=G(x,t)+\nu\frac{\partial^{2}\theta}{\partial x_{1}^{2}}(x_{1},t)-\frac{\partial\theta}{\partial t}(x_{1},t)-u^{1}(x,t)\frac{\partial\theta}{\partial x_{1}}(x_{1},t).\label{G-theta}
\end{equation}
\end{thm}

\begin{proof}
The proof of this theorem is based on the representation (\ref{key-m01})
and (\ref{M-1}) by sending $\varepsilon\uparrow\infty$. Let $\varepsilon\uparrow\infty$
to obtain (with $x=(x_{1},x_{2})$) 
\begin{equation}
\lim_{\varepsilon\uparrow\infty}W_{0}^{\varepsilon}(x)=W_{0}(x)\equiv\omega_{0}(x)+\frac{\partial u_{0}^{1}}{\partial x_{2}}(x_{1},0),\label{in-it-data-01}
\end{equation}
\begin{equation}
\lim_{\varepsilon\uparrow\infty}\sigma_{\varepsilon}(x_{1},x_{2},t)=\theta(x_{1},t)\label{b-v-01}
\end{equation}
and 
\begin{equation}
\lim_{\varepsilon\downarrow0}g_{\varepsilon}(x,t)=g(x,t).\label{b-v-02}
\end{equation}
Hence the conclusion follows immediately. 
\end{proof}
The following stochastic representation provides another approach
which avoids the discussion of the dynamics of the boundary vorticity
$\theta$ and therefore it has some advantage when the viscosity $\nu>0$
is small. 
\begin{thm}
\label{thm:main} Let $u(x,t)$ be a solution to Navier-Stokes equations
(\ref{meq-01}, \ref{meq-02}) in $D=\{x:x_{2}<0\}$, and $G=\nabla\wedge F$.
Assume that both $u(x,t)$ and $G(x,t)$ have twice continuous derivatives
on $\overline{D}$, and assume that $u$ and $G$ and their derivatives
are integrable on $D$. Then $u(x,t)$ and the Taylor diffusion with
its infinitesimal generator $\nu\Delta+u\cdot\nabla$ form the closed
random vortex dynamics 
\begin{equation}
\begin{cases}
u(x,t)=\int_{D}\mathbb{E}\left[1_{D}(X_{t}^{\xi,0})K(x,X_{t}^{\xi,0})-1_{D}(X_{t}^{\bar{\xi},0})K(x,X_{t}^{\bar{\xi},0})\right]\omega_{0}(\xi)\textrm{d}\xi\\
\quad\quad\quad+\int_{0}^{t}\int_{D}\mathbb{E}\left[1_{D}(X_{t}^{\xi,s})K(x,X_{t}^{\xi,s})-1_{D}(X_{t}^{\bar{\xi},s})K(x,X_{t}^{\bar{\xi},s})\right]G(\xi,s)\textrm{d}\xi\textrm{d}s\\
\quad\quad\quad-2\nu\int_{\mathbb{R}}\int_{0}^{t}\left.\frac{\partial}{\partial\xi_{2}}\right|_{\xi_{2}=0}\mathbb{E}\left[1_{D}(X_{t}^{\xi,s})K(x,X_{t}^{\xi,s})\right]\theta(\xi_{1},s)\textrm{d}\xi_{1}\textrm{d}s, & \textrm{ for }x_{2}<0,\\
u(x,t)=\overline{u(\bar{x},t)}, & \textrm{ for }x_{2}>0,\\
\textrm{d}X_{t}^{\xi,s}=u(X_{t}^{\xi,s},t)\textrm{d}t+\sqrt{2\nu}\textrm{d}B_{t},\quad X_{s}^{\xi,s}=\xi, & \textrm{ for }\xi\in\mathbb{R}^{2},
\end{cases}\label{M-1}
\end{equation}
where $\omega_{0}(\xi)=\nabla\wedge u_{0}(x)$ and $u_{0}(x)=u(x,0)$
are the initial data. 
\end{thm}

\begin{proof}
The proof of this theorem is based on the representation (\ref{key-m01})
as well, and the stochastic representation (\ref{M-1}) by sending
$\varepsilon\downarrow0$. Let us take the following explicit cut-off
function defined by 
\begin{equation}
\phi(r)=\begin{cases}
1 & \textrm{for \ensuremath{r\in[0,1/3)},}\\
\frac{1}{2}+54\left(r-\frac{1}{2}\right)^{3}-\frac{9}{2}\left(r-\frac{1}{2}\right) & \textrm{ for }r\in[1/3,2/3],\\
0 & \textrm{ for }r\geq2/3
\end{cases}\label{phi-def}
\end{equation}
Then $-54\leq\phi''\leq54$, $-\frac{9}{2}\leq\phi'\leq0$ on $[1/3,2/3]$
and $\phi'=0$ for $r\leq1/3$ or $r\geq2/3$. In fact 
\begin{equation}
\phi'(r)=\begin{cases}
162\left(r-\frac{1}{2}\right)^{2}-\frac{9}{2} & \textrm{ for }r\in[1/3,2/3],\\
0 & \textrm{ otherwise }
\end{cases}\label{phi-1d}
\end{equation}
and 
\begin{equation}
\phi''(r)=\begin{cases}
324\left(r-\frac{1}{2}\right) & \textrm{ for }r\in[1/3,2/3],\\
0 & \textrm{ otherwise. }
\end{cases}\label{phi-2d}
\end{equation}
The key is to show that the last term on the right-hand side of (\ref{key-m01}),
that is, 
\[
E_{\varepsilon}(x,t)\equiv\int_{D}\int_{0}^{t}\int_{D}K(x,y)\left(p_{u}(s,\xi,t,y)-p_{u}(s,\bar{\xi},t,y)\right)g_{\varepsilon}(\xi,s)\textrm{d}\xi\textrm{d}s\textrm{d}y
\]
has a limit as $\varepsilon\downarrow0$, where $g_{\varepsilon}$
is given by (\ref{g-xt-01}). To this end we consider the integral
\[
Q^{\varepsilon}(s)=\int_{D}\left(p_{u}(s,\xi,t,y)-p_{u}(s,\bar{\xi},t,y)\right)g_{\varepsilon}(\xi,s)\textrm{d}\xi.
\]
Since $\phi(-x_{2}/\varepsilon)\rightarrow1_{\{0\}}(x_{2})$ as $\varepsilon\downarrow0$,
so that there are no contributions towards $Q^{\varepsilon}(s)$ from
the last two terms on the right-hand side (\ref{g-xt-01}). In fact
the sum of the two terms 
\[
\phi(-x_{2}/\varepsilon)\left(\nu\frac{\partial^{2}\theta}{\partial x_{1}^{2}}(x_{1},t)-\frac{\partial\theta}{\partial t}(x_{1},t)\right)-\phi(-x_{2}/\varepsilon)u^{1}(x,t)\frac{\partial\theta}{\partial x_{1}}(x_{1},t)
\]
tends to, as $u^{1}(x,t)=0$ on $x_{2}=0$, 
\[
1_{\{0\}}(x_{2})\left(\nu\frac{\partial^{2}\theta}{\partial x_{1}^{2}}(x_{1},t)-\frac{\partial\theta}{\partial t}(x_{1},t)\right)
\]
as $\varepsilon\downarrow0$, which equals to zero almost surely on
$D$. Let us handle the singular terms involving $\frac{1}{\varepsilon}$
in $Q^{\varepsilon}(s)$. There are two integrals we need to consider:
\[
Q_{1}^{\varepsilon}(s)=\int_{D}\left(p_{u}(s,\xi,t,y)-p_{u}(s,\bar{\xi},t,y)\right)\frac{1}{\varepsilon}\phi'(-\xi_{2}/\varepsilon)u^{2}(\xi,s)\theta(\xi_{1},s)\textrm{d}\xi
\]
and 
\[
Q_{2}^{\varepsilon}(s)=\int_{D}\left(p_{u}(s,\xi,t,y)-p_{u}(s,\bar{\xi},t,y)\right)\frac{\nu}{\varepsilon^{2}}\phi''(-\xi_{2}/\varepsilon)\theta(\xi_{1},s)\textrm{d}\xi.
\]
We want to find their limits as $\varepsilon\downarrow0$. Suppose
$\beta$ is a smooth function on $\mathbb{R}^{2}$ with a compact
support, then 
\begin{align*}
\int_{\mathbb{R}^{2}}\beta(x)\phi'(-x_{2}/\varepsilon)\textrm{d}x & =\varepsilon\int_{-\infty}^{\infty}\int_{1/3}^{2/3}\beta(x_{1},-\varepsilon x_{2})\phi'(x_{2})\textrm{d}x_{2}\textrm{d}x_{1}\\
 & =\varepsilon^{2}\int_{\mathbb{R}}\int_{1/3}^{2/3}\frac{\partial\beta}{\partial x_{2}}(x_{1},-\varepsilon x_{2})\phi(x_{2})\textrm{d}x_{2}\textrm{d}x_{1}\\
 & +\varepsilon\int_{\mathbb{R}}\left[-\beta(x_{1},-\varepsilon/3)\right]dx_{1}
\end{align*}
so that 
\begin{equation}
\lim_{\varepsilon\downarrow0}\frac{1}{\varepsilon}\int_{\mathbb{R}^{2}}\beta(x)\phi'(-x_{2}/\varepsilon)\textrm{d}x=-\int_{\mathbb{R}}\beta(x_{1},0)\textrm{d}x_{1}.\label{w-1st-01}
\end{equation}
In other words, 
\begin{equation}
\lim_{\varepsilon\downarrow0}\frac{1}{\varepsilon}\phi'(-x_{2}/\varepsilon)\textrm{d}x_{1}\textrm{d}x_{2}=-\textrm{d}x_{1}\delta_{0}(\textrm{d}x_{2})\quad\textrm{ as }\varepsilon\rightarrow0.\label{w-1st-02}
\end{equation}
Hence 
\begin{align}
\lim_{\varepsilon\downarrow0}Q_{1}^{\varepsilon}(s) & =\lim_{\varepsilon\downarrow0}\int_{D}\left(p_{u}(s,\xi,t,y)-p_{u}(s,\bar{\xi},t,y)\right)\frac{1}{\varepsilon}\phi'(-\xi_{2}/\varepsilon)u^{2}(\xi,t)\theta(\xi_{1},t)\textrm{d}\xi\nonumber \\
 & =-\int_{\mathbb{R}}\left.\left(p_{u}(s,\xi,t,y)-p_{u}(s,\bar{\xi},t,y)\right)u^{2}(\xi,t)\right|_{\xi_{2}=0}\theta(\xi_{1},t)\textrm{d}\xi_{1}\nonumber \\
 & =0.\label{Q1-e1}
\end{align}
Similarly, since 
\begin{align*}
\int_{\mathbb{R}^{2}}\beta(x)\phi''(-x_{2}/\varepsilon)\textrm{d}x & =\varepsilon\int_{\mathbb{R}}\int_{1/3}^{2/3}\beta(x_{1},-\varepsilon x_{2})\phi''(x_{2})\textrm{d}x_{2}\textrm{d}x_{1}\\
 & =\varepsilon^{2}\int_{\mathbb{R}}\int_{1/3}^{2/3}\frac{\partial\beta}{\partial x_{2}}(x_{1},-\varepsilon x_{2})\phi'(x_{2})\textrm{d}x_{2}\textrm{d}x_{1}\\
 & =\varepsilon^{3}\int_{\mathbb{R}}\int_{1/3}^{2/3}\frac{\partial^{2}\beta}{\partial x_{2}^{2}}(x_{1},-\varepsilon x_{2})\phi(x_{2})\textrm{d}x_{2}\textrm{d}x_{1}\\
 & -\varepsilon^{2}\int_{\mathbb{R}}\frac{\partial\beta}{\partial x_{2}}(x_{1},-\varepsilon/3)\textrm{d}x_{1},
\end{align*}
so that 
\begin{equation}
\lim_{\varepsilon\downarrow0}\frac{1}{\varepsilon^{2}}\int_{\mathbb{R}^{2}}\beta(x)\phi''(-x_{2}/\varepsilon)\textrm{d}x=-\int_{\mathbb{R}}\frac{\partial\beta}{\partial x_{2}}(x_{1},0)\textrm{d}x_{1}.\label{2nd-phi-01}
\end{equation}
Therefore 
\begin{align}
\lim_{\varepsilon\downarrow0}Q_{2}^{\varepsilon}(s) & =\lim_{\varepsilon\downarrow0}\int_{D}\left(p_{u}(s,\xi,t,y)-p_{u}(s,\bar{\xi},t,y)\right)\frac{\nu}{\varepsilon^{2}}\phi''(-\xi_{2}/\varepsilon)\theta(\xi_{1},s)\textrm{d}\xi\nonumber \\
 & =\lim_{\varepsilon\downarrow0}\nu\int_{D}\left(p_{u}(s,\xi,t,y)-p_{u}(s,\bar{\xi},t,y)\right)\theta(\xi_{1},s)\frac{1}{\varepsilon^{2}}\phi''(-\xi_{2}/\varepsilon)\textrm{d}\xi\nonumber \\
 & =-2\nu\int_{\mathbb{R}}\left.\frac{\partial}{\partial\xi_{2}}\right|_{\xi_{2}=0}p_{u}(s,(\xi_{1},\xi_{2}),t,y)\theta(\xi_{1},s)\textrm{d}\xi_{1}.\label{Q2-e2}
\end{align}
Limits (\ref{Q1-e1}) and (\ref{Q2-e2}) together imply that 
\begin{align*}
\lim_{\varepsilon\downarrow0}Q^{\varepsilon}(s) & =\lim_{\varepsilon\downarrow0}\int_{D}\left(p_{u}(s,\xi,t,y)-p_{u}(s,\bar{\xi},t,y)\right)g_{\varepsilon}(\xi,s)\textrm{d}\xi\\
 & =\int_{D}\left(p_{u}(s,\xi,t,y)-p_{u}(s,\bar{\xi},t,y)\right)G(\xi,s)\textrm{d}\xi\\
 & -2\nu\int_{\mathbb{R}}\left.\frac{\partial}{\partial\xi_{2}}\right|_{\xi_{2}=0}p_{u}(s,(\xi_{1},\xi_{2}),t,y)\theta(\xi_{1},s)\textrm{d}\xi_{1},
\end{align*}
and it in turn yields that 
\begin{align*}
\lim_{\varepsilon\downarrow0}E_{\varepsilon}(x,t) & =\int_{D}\int_{0}^{t}\int_{D}K(x,y)\left(p_{u}(s,\xi,t,y)-p_{u}(s,\bar{\xi},t,y)\right)G(\xi,s)\textrm{d}\xi\textrm{d}s\textrm{d}y\\
 & -2\nu\int_{\mathbb{R}}\int_{0}^{t}\int_{D}K(x,y)\left.\frac{\partial}{\partial\xi_{2}}\right|_{\xi_{2}=0}p_{u}(s,(\xi_{1},\xi_{2}),t,y)\theta(\xi_{1},s)\textrm{d}\xi_{1}\textrm{d}s\textrm{d}y\\
 & =\int_{0}^{t}\int_{D}\mathbb{E}\left[1_{D}(X_{t}^{\xi,s})K(x,X_{t}^{\xi,s})-1_{D}(X_{t}^{\bar{\xi},s})K(x,X_{t}^{\bar{\xi},s})\right]G(\xi,s)\textrm{d}\xi\textrm{d}s\\
 & -2\nu\int_{\mathbb{R}}\int_{0}^{t}\left.\frac{\partial}{\partial\xi_{2}}\right|_{\xi_{2}=0}\mathbb{E}\left[1_{D}(X_{t}^{\xi,s})K(x,X_{t}^{\xi,s})\right]\theta(\xi_{1},s)\textrm{d}\xi_{1}\textrm{d}s.
\end{align*}

We then deal with other two terms in (\ref{key-m01}). Since 
\begin{equation}
W_{0}^{\varepsilon}(x)=\omega_{0}(x)+\phi(-x_{2}/\varepsilon)\frac{\partial u_{0}}{\partial x_{2}}(x_{1},0)\label{init-01}
\end{equation}
for $x=(x_{1},x_{2})$ with $x_{2}\leq0$, it follows easily that
\begin{equation}
\lim_{\varepsilon\downarrow0}W_{0}^{\varepsilon}(x)=\omega_{0}(x_{1},x_{2})+1_{\{0\}}(x_{2})\frac{\partial u_{0}}{\partial x_{2}}(x_{1},0)\equiv W_{0}(x),\label{W0-refor}
\end{equation}
which allows to replace $W_{0}^{\varepsilon}$ by $\omega_{0}$ in
the first term on the right-hand side of (\ref{key-m01}).

In the next step we show the limit of the second term on the right-hand
side (\ref{key-m01}) vanishes as $\varepsilon\downarrow0$. Hence
we need to handle the integrals 
\[
I_{\varepsilon}^{i}(x)=\int_{D}K^{i}(x,y)\sigma_{\varepsilon}(y,t)\textrm{d}y
\]
where $i=1,2$. By the definition of the cut-off function $\phi$,
\[
\int_{D}K(x,y)\sigma_{\varepsilon}(y,t)\textrm{d}y=\varepsilon\int_{-\infty}^{\infty}\left(\theta(y_{1},t)\int_{0}^{2/3}K(x,(y_{1},-\varepsilon y_{2}))\phi(y_{2})\textrm{d}y_{2}\right)\textrm{d}y_{1}.
\]
Since the integral against $y_{2}$ can be calculated as the following:
\begin{align*}
J_{\varepsilon}^{1}(x;y_{1}) & \equiv\int_{0}^{2/3}K^{1}(x,(y_{1},-\varepsilon y_{2}))\phi(y_{2})\textrm{d}y_{2}\\
 & =\frac{1}{2\pi}\int_{0}^{2/3}\left(\frac{-\varepsilon y_{2}-x_{2}}{(y_{1}-x_{1})^{2}+(x_{2}+\varepsilon y_{2})^{2}}-\frac{-\varepsilon y_{2}+x_{2}}{(y_{1}-x_{1})^{2}+(x_{2}-\varepsilon y_{2})^{2}}\right)\phi(y_{2})\textrm{d}y_{2}\\
 & \rightarrow\frac{1}{\pi}\frac{x_{2}}{(y_{1}-x_{1})^{2}+(x_{2})^{2}}\int_{0}^{2/3}\phi(y_{2})\textrm{d}y_{2}
\end{align*}
as $\varepsilon\downarrow0$, and similarly 
\begin{align*}
J_{\varepsilon}^{2}(x;y_{1}) & =\int_{0}^{2/3}K^{2}(x,(y_{1},-\varepsilon y_{2}))\phi(y_{2})\textrm{d}y_{2}\\
 & =\frac{1}{2\pi}(y_{1}-x_{1})\int_{0}^{2/3}\left(\frac{1}{(y_{1}-x_{1})^{2}+(x_{2}-\varepsilon y_{2})^{2}}-\frac{1}{(y_{1}-x_{1})^{2}+(x_{2}+\varepsilon y_{2})^{2}}\right)\phi(y_{2})\textrm{d}y_{2}\\
 & \rightarrow0
\end{align*}
as $\varepsilon\downarrow0$. Therefore 
\[
\lim_{\varepsilon\downarrow0}\int_{D}K(x,y)\sigma_{\varepsilon}(y,t)\textrm{d}y=0
\]
which allows us to drop the second term in (\ref{key-m01}) as $\varepsilon\downarrow0$.
The conclusion of the theorem now follows immediately. 
\end{proof}

\section{Monte-Carlo simulations}

We retain the assumptions in Section \ref{Sec33} and we assume that
the technical conditions in Theorem \ref{thm:main} are satisfied.
To implement Monte-Carlo simulations via the stochastic formulation,
Theorem \ref{thm:g}, we appeal to the ideas from random vortex method.
That is, either dropping the mathematical expectations by running
independent Brownian motions, or appealing to the strong law of large
numbers so that the expectations are replaced by running a number
of independent copies of the Taylor diffusion. Let us describe the
methods in more details.

Recall that we have established the following representation 
\begin{equation}
\begin{cases}
u(x,t)=\int_{D}\mathbb{E}\left[1_{D}(X_{t}^{\xi,0})K(x,X_{t}^{\xi,0})-1_{D}(X_{t}^{\bar{\xi},0})K(x,X_{t}^{\bar{\xi},0})\right]\omega_{0}(\xi)\textrm{d}\xi\\
\quad\quad\quad+\int_{D}\int_{0}^{t}\left[1_{D}(X_{t}^{\xi,s})K(x,X_{t}^{\xi,s})-1_{D}(X_{t}^{\bar{\xi},s})K(x,X_{t}^{\bar{\xi},s})\right]G(\xi,s)\textrm{d}\xi_{1}\textrm{d}s\\
\quad\quad\quad-2\nu\int_{\mathbb{R}}\int_{0}^{t}\left.\frac{\partial}{\partial\xi_{2}}\right|_{\xi_{2}=0}\mathbb{E}\left[1_{D}(X_{t}^{\xi,s})K(x,X_{t}^{\xi,s})\right]\theta(\xi_{1},s)\textrm{d}\xi_{1}\textrm{d}s, & \textrm{ for }x_{2}<0,\\
u(x,t)=\overline{u(\bar{x},t)}, & \textrm{ for }x_{2}>0,\\
X_{t}^{\xi,s}=\xi+\int_{s}^{t}u(X_{r}^{\xi,s},r)\textrm{d}r+\sqrt{2\nu}(B_{t}-B_{s}),\quad\textrm{ for }\xi\in\mathbb{R}^{2}\textrm{ and for }s\geq0, & \textrm{ for }\xi\in\mathbb{R}^{2},
\end{cases}\label{VDS-01}
\end{equation}
where $\omega_{0}(\xi)=\nabla\wedge u_{0}(x)$ and $B$ is a two dimensional
Brownian motion. In numerical schemes described below the boundary
vorticity $\theta$ has to be updated over the time during iterations,
rather than through modelling or boundary layer equations. However,
for small viscosity $\nu$, the term involving the boundary vorticity
may be dropped, and for the case where there is no external force,
then the previous system can be approximated by the following simpler
random vortex dynamics 
\begin{equation}
\begin{cases}
\tilde{u}(x,t)=\int_{D}\mathbb{E}\left[1_{D}(X_{t}^{\xi})K(x,X_{t}^{\xi})-1_{D}(X_{t}^{\bar{\xi}})K(x,X_{t}^{\bar{\xi}})\right]\omega_{0}(\xi)\textrm{d}\xi, & \textrm{ for }x_{2}<0,\\
\tilde{u}(x,t)=\overline{\tilde{u}(\bar{x},t)}, & \textrm{ for }x_{2}>0,\\
X_{t}^{\xi}=\xi+\int_{0}^{t}u(X_{r}^{\xi},r)\textrm{d}r+\sqrt{2\nu}B_{t},\quad\textrm{ for }\xi\in\mathbb{R}^{2}, & \textrm{ for }\xi\in\mathbb{R}^{2}.
\end{cases}\label{M-add}
\end{equation}
The use of this approximation will reduce the computing cost, and
therefore more computing hours may be saved for performing simulations
of boundary turbulent flows where the viscosity is small while demands
for finer scales.

Let us describe the numerical schemes we are going to perform the
numerical experiments. We divide the schemes into two steps. During
the first step we set up the discretization procedure for dealing
with (finite dimensional) integrals in time and in the space variables.
The discretization of this type appears in any numerical methods,
but a bit care is needed due to the appearance of boundary layer phenomena
for wall-bounded flows. According to Prandtl \citep{Prandtl1904},
there is a thin layer near the boundary within which the main stream
velocity decreases to zero sharply, hence there is substantial stress
at the wall, which in turn generates significant boundary vorticity.
Turbulence may be generated near the solid wall if the Reynolds number
is large. The boundary layer thickness $\delta>0$ of a fluid flow
with small viscosity is given by $\frac{\delta}{L}\sim\sqrt{\frac{1}{\textrm{Re}}}$,
where $\textrm{Re}$ is the Reynolds number $\textrm{Re}=\frac{VL}{\nu}$,
$V$ and $L$ are the typical velocity magnitude and typical length.
In numerical simulations, we can always first make a reduction, to
make the fluid dynamics equations dimensionless. That is, the typical
velocity and length may be fixed to the good size for printing the
outcomes. While we still prefer to use the Reynolds number $\textrm{Re}$,
so that $\nu\sim\frac{1}{\textrm{Re}}$. Suppose $\nu$ is small,
then near the boundary, the mesh $h_{2}$ for the $x_{2}$-coordinate
has to be far smaller than the boundary layer thickness $\delta\sim\sqrt{\frac{1}{\textrm{Re}}}$,
which leads to the first constraint: 
\begin{equation}
h_{2}\ll L\sqrt{\frac{1}{\textrm{Re}}}.\label{h-ratio01}
\end{equation}
The mesh size $h_{1}$ for $x_{1}$-coordinate should be comparable
to $h_{2}$ but not need to be larger than $h_{2}$. For the region
outside the boundary layer, we can use reasonable size of the mesh
$h_{0}$, and in general we choose $h_{0}\geq h_{1}\geq h_{2}$. We
may choose integers $N_{1}$, $N_{2}$ and $N_{0}$ according to the
following constraints: 
\begin{equation}
N_{1}h_{1}\sim L,\quad N_{2}h_{2}\geq\delta\quad\textrm{ and }N_{0}h_{0}\sim L.\label{meshes}
\end{equation}

We calculate the values at the lattice points 
\begin{equation}
x^{i_{1},i_{2}}=\begin{cases}
(i_{1}h_{1},i_{2}h_{2}) & \textrm{if }|i_{1}|\leq N_{1}\textrm{ and }|i_{2}|\leq N_{2};\\
(i_{1}h_{0},i_{2}h_{0}) & \textrm{if }|i_{1}|\leq N_{0}\textrm{ and }N_{2}<|i_{2}|\leq N_{0}+N_{2}.
\end{cases}\label{xii-01}
\end{equation}
That is, $x^{i_{1},i_{2}}=(x_{1}^{i_{1},i_{2}},x_{2}^{i_{1},i_{2}})$
where 
\begin{equation}
\begin{cases}
x_{1}^{i_{1},i_{2}}=i_{1}h_{1}\textrm{ and }x_{2}^{i_{1},i_{2}}=i_{2}h_{2} & \textrm{if }|i_{1}|\leq N_{1}\textrm{ and }|i_{2}|\leq N_{2};\\
x_{1}^{i_{1},i_{2}}=i_{1}h_{0}\textrm{ and }x_{2}^{i_{1},i_{2}}=i_{2}h_{0} & \textrm{if }|i_{1}|\leq N_{0}\textrm{ and }N_{2}<|i_{2}|\leq N_{0}+N_{2}.
\end{cases}\label{xi1i2-01}
\end{equation}
The total number of the lattice points we require is given by 
\[
n=2(N_{1}+1)(N_{2}+1)+2(N_{0}+1)^{2}
\]
which largely determine the computational cost. Thus far we have described
the discretization of space variables.

For time $t$, we can use a unified scheme, at each step we use time
duration $h>0$, to be chosen properly.

The Taylor diffusion (\ref{X-sde1}) has to be issued from the lattice
points in the numerical scheme at starting time $t_{0}=0$, $t_{1}=h$,
and so on $t_{k}=kh$ for $k=0,1,2,\cdots$. A simple scheme, but
it is not claimed to be the best one, for the discretization of the
SDE can be set up as the following: 
\[
X_{t_{l+1}}^{i_{1},i_{2};t_{l}}=x^{i_{1},i_{2}}+h\tilde{u}(x^{i_{1},i_{2}},0)+\sqrt{2\nu}(B_{t_{l+1}}-B_{t_{l}})
\]
(for $l=0,1,2,\cdots$), then 
\[
X_{t_{k+1}}^{i_{1},i_{2};t_{l}}=X_{t_{k}}^{i_{1},i_{2};t_{l}}+h\tilde{u}(X_{t_{k}}^{i_{1},i_{2};t_{l}},t_{k})+\sqrt{2\nu}(B_{t_{k+1}}-B_{t_{k}})
\]
for $k=l,l+1,\cdots$, where 
\[
\tilde{u}(x,t_{0})=u_{0}(x),
\]
\begin{align}
\tilde{u}(x,t_{k+1}) & =\sum_{i_{1},i_{2}}A_{i_{1},i_{2}}\omega_{i_{1},i_{2}}\mathbb{E}\left[K_{D}(x,X_{t_{k}}^{i_{1},i_{2};0})-K_{D}(x,X_{t_{k}}^{i_{1},-i_{2};0})\right]\nonumber \\
 & +\sum_{i_{1},i_{2}}A_{i_{1},i_{2}}\sum_{l=0}^{k}hG_{i_{1},i_{2};k}\mathbb{E}\left[K_{D}(x,X_{t_{k}}^{i_{1},i_{2};t_{l}})-K_{D}(x,X_{t_{k}}^{i_{1},-i_{2};t_{l}})\right]\nonumber \\
 & +2\nu\sum_{i_{1}}\frac{h_{1}}{h_{2}}\sum_{l=0}^{k}h\left(\mathbb{E}\left[K_{D}(x,X_{t_{k}}^{i_{1},-h_{2};t_{l}})\right]-\mathbb{E}\left[K_{D}(x,X_{t_{k}}^{i_{1},0;t_{l}})\right]\right)\theta_{i_{1};l},\label{g-s-01}
\end{align}
\[
\theta_{i_{1};l}=\frac{\tilde{u}^{1}((i_{1}h_{1},-h_{2}),t_{l})}{h_{2}},\quad\textrm{ for }l=0,\ldots,k,
\]
for $x_{2}<0$, and 
\[
(\tilde{u}^{1}(x,t),\tilde{u}^{2}(x,t))=(\tilde{u}^{1}(x,t),-\tilde{u}^{2}(\bar{x},t))\quad\textrm{ for }x_{2}>0.
\]
which is solved for $i_{1},i_{2}$ within the regions assigned, where
\begin{equation}
A_{i_{1},i_{2}}=\begin{cases}
h_{1}h_{2} & \textrm{if }|i_{1}|\leq N_{1}\textrm{ and }|i_{2}|\leq N_{2};\\
h_{0}h_{0} & \textrm{if }|i_{1}|\leq N_{0}\textrm{ and }N_{2}<|i_{2}|\leq N_{0}+N_{2},
\end{cases}\label{Aii-01}
\end{equation}
and 
\begin{equation}
\omega_{i_{1},i_{2}}=\omega_{0}(x^{i_{1},i_{2}}),\quad G_{i_{1},i_{2};k}=G(x^{i_{1},i_{2}},t_{k})\label{ome-d}
\end{equation}
and $\tilde{u}(x,0)=u_{0}(x)$. Note that $\tilde{u}=(\tilde{u}^{1},\tilde{u}^{2})$.
Also here for simplicity, we have introduce the notation that 
\begin{equation}
K_{D}(x,y)=\begin{cases}
K(x,y) & \textrm{ if }y\in D,\\
0 & \textrm{ if }y\notin D.
\end{cases}\label{KD-01-1}
\end{equation}

In the second step, one has to handle the mathematical expectations,
which is the core of the Monte-Carlo schemes. Slightly different approaches
lead to the following schemes.

\subsection{Numerical scheme 1}

We are now in a position to formulate our first Monte-Carlo scheme
based on the simple random vortex dynamic (\ref{M-add}). In this
scheme, we drop the expectation by running independent (two dimensional)
Brownian motions. More precisely we run the following stochastic differential
equations: 
\[
X_{t_{l+1}}^{i_{1},i_{2};t_{l}}=x^{i_{1},i_{2}}+h\tilde{u}(x^{i_{1},i_{2}},0)+\sqrt{2\nu}(B_{t_{l+1}}^{i_{1},i_{2}}-B_{t_{l}}^{i_{1},i_{2}})
\]
(for $l=0,1,2,\cdots$), then 
\[
X_{t_{k+1}}^{i_{1},i_{2};t_{l}}=X_{t_{k}}^{i_{1},i_{2};t_{l}}+h\tilde{u}(X_{t_{k}}^{i_{1},i_{2};t_{l}},t_{k})+\sqrt{2\nu}(B_{t_{k+1}}^{i_{1},i_{2}}-B_{t_{k}}^{i_{1},i_{2}})
\]
for $k=l,l+1,\cdots$, and 
\[
\tilde{u}(x,0)=u_{0}(x),
\]
\begin{align}
\tilde{u}(x,t_{k+1}) & =\sum_{i_{1},i_{2}}A_{i_{1},i_{2}}\omega_{i_{1},i_{2}}\left[K_{D}(x,X_{t_{k}}^{i_{1},i_{2};0})-K_{D}(x,X_{t_{k}}^{i_{1},-i_{2};0})\right]\nonumber \\
 & +\sum_{i_{1},i_{2}}A_{i_{1},i_{2}}\sum_{l=0}^{k}hG_{i_{1},i_{2};k}\left[K_{D}(x,X_{t_{k}}^{i_{1},i_{2};t_{l}})-K_{D}(x,X_{t_{k}}^{i_{1},-i_{2};t_{l}})\right]\nonumber \\
 & +2\nu\sum_{i_{1}}\frac{h_{1}}{h_{2}}\sum_{l=0}^{k}h\left(\left[K_{D}(x,X_{t_{k}}^{i_{1},-h_{2};t_{l}})\right]-\mathbb{E}\left[K_{D}(x,X_{t_{k}}^{i_{1},0;t_{l}})\right]\right)\theta_{i_{1};l},\label{M-3}
\end{align}
\[
\theta_{i_{1};l}=\frac{\tilde{u}^{1}((i_{1}h_{1},-h_{2}),t_{l})}{h_{2}},\qquad\textrm{ for }l=0,\ldots,k,
\]
where $A_{i_{1},i_{2}}$ , $\omega_{i_{1},i_{2}}$ are given (\ref{Aii-01})
and (\ref{ome-d}), and $B^{i_{1},i_{2}}$ (where $(i_{1},i_{2})$
runs through the assigned indices) is a family of independent two
dimensional Brownian motions.

The convergence of this scheme of course can not be taken as granted
which itself is an important mathematical problem, for the case without
surface boundary, the convergence has been proved in \citep{AndersonGreengard1985},
\citep{Long1988} and \citep{Goodman1987}. We will investigate this
problem for surface boundary case in a separate work.

\subsection{Numerical scheme 2}

In this scheme we appeal to the strong law of large numbers, that
is we replace the expectation by its average. Therefore we solve the
following stochastic differential equations: 
\[
X_{t_{l+1}}^{i_{1},i_{2};t_{l}}=x^{i_{1},i_{2}}+h\tilde{u}(x^{i_{1},i_{2}},0)+\sqrt{2\nu}(B_{t_{l+1}}^{m}-B_{t_{l}}^{m})
\]
(for $l=0,1,2,\cdots$), then 
\[
X_{t_{k+1}}^{m;i_{1},i_{2};t_{l}}=X_{t_{k}}^{m;i_{1},i_{2};t_{l}}+h\tilde{u}(X_{t_{k}}^{m;i_{1},i_{2};t_{l}},t_{k})+\sqrt{2\nu}(B_{t_{k+1}}^{m}-B_{t_{k}}^{m})
\]
for $k=l,l+1,\cdots$, $m=1,\ldots,N$, and 
\[
\tilde{u}(x,0)=u_{0}(x),
\]
\begin{align}
\tilde{u}(x,t_{k+1}) & =\sum_{i_{1},i_{2}}A_{i_{1},i_{2}}\omega_{i_{1},i_{2}}\frac{1}{N}\sum_{m=1}^{N}\left[K_{D}(x,X_{t_{k}}^{m;i_{1},i_{2};0})-K_{D}(x,X_{t_{k}}^{m;i_{1},-i_{2};0})\right]\nonumber \\
 & +\sum_{i_{1},i_{2}}A_{i_{1},i_{2}}\sum_{l=0}^{k}hG_{i_{1},i_{2};k}\frac{1}{N}\sum_{m=1}^{N}\left[K_{D}(x,X_{t_{k}}^{m;i_{1},i_{2};t_{l}})-K_{D}(x,X_{t_{k}}^{m;i_{1},-i_{2};t_{l}})\right]\nonumber \\
 & +2\nu\sum_{i_{1}}\frac{h_{1}}{h_{2}}\sum_{l=0}^{k}h\frac{1}{N}\sum_{m=1}^{N}\left[K_{D}(x,X_{t_{k}}^{m;i_{1},-h_{2};t_{l}})-K_{D}(x,X_{t_{k}}^{m;i_{1},0;t_{l}})\right]\theta_{i_{1};l},\label{M-3-1}
\end{align}
\[
\theta_{i_{1};l}=\frac{\tilde{u}^{1}((i_{1}h_{1},-h_{2}),t_{l})}{h_{2}}\quad\textrm{ for }l=0,\cdots,k
\]
The data $A_{i_{1},i_{2}}$, $\omega_{i_{1},i_{2}}$ and so on are
the same as in the previous scheme.

\subsection{Numerical experiment}

Since the aim of the paper is to report a theoretical framework for
performing Monte-Carlo simulations for incompressible fluid flows
passing over a solid wall, to demonstrate the usefulness of the method,
we report a simple numerical experiment for the fluid flow with small
viscosity $\nu>0$, for which the updating of the boundary stress
$\theta$ seems not necessary, and therefore less demanding in computing
power. We hope to improve the numerical experiment results in future
work. 

Recall the Reynolds number is given by $\textrm{Re}=\frac{U_{0}L}{\nu}$
where $U_{0}$ the main stream velocity and $L$ is a typical size,
which is determined by the initial velocity in our experiments. For
simplicity, in our experiment, we set 
\begin{equation}
\nu=0.15.\label{viscosity-c1}
\end{equation}
The typical length scale 
\begin{equation}
L=2\pi.\label{typ-s-01}
\end{equation}
This gives the initial main stream velocity 
\begin{equation}
U_{0}=\frac{\nu}{L}\textrm{Re}=0.02388535\times\textrm{Re}.\label{k-value}
\end{equation}
The initial velocity in our experiments is chosen as the following:
\[
u_{0}(x_{1},x_{2})=(-U_{0}\sin x_{2},0)
\]
so that the initial vorticity 
\begin{equation}
\omega_{0}(x_{1},x_{2})=U_{0}\cos x_{2}.\label{int-vort-01}
\end{equation}

In the numerical experiment, we choose $\textrm{Re}=2500$. Then $U_{0}=\alpha=59.7$.
Choose 
\[
h_{2}\ll\delta=L\sqrt{\frac{1}{\textrm{Re}}}=0.1256.
\]
Let $h_{1}=0.3$, and $h_{0}=0.4$. The numerical result is demonstrated
in the figure.

\begin{figure}
\centering %
\begin{tabular}{cc}
\includegraphics[width=4cm,height=3.2cm]{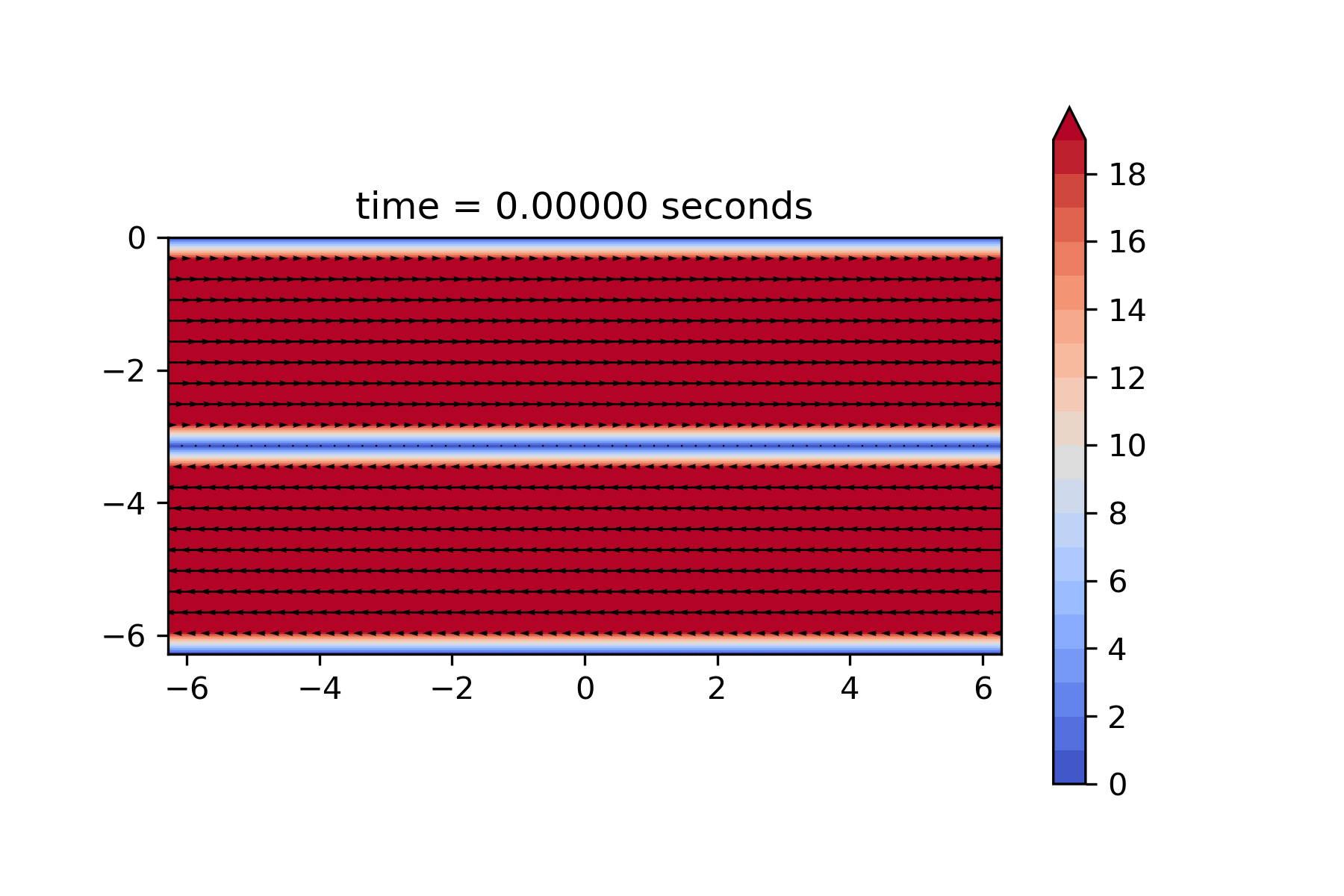}  & \includegraphics[width=4cm,height=3.2cm]{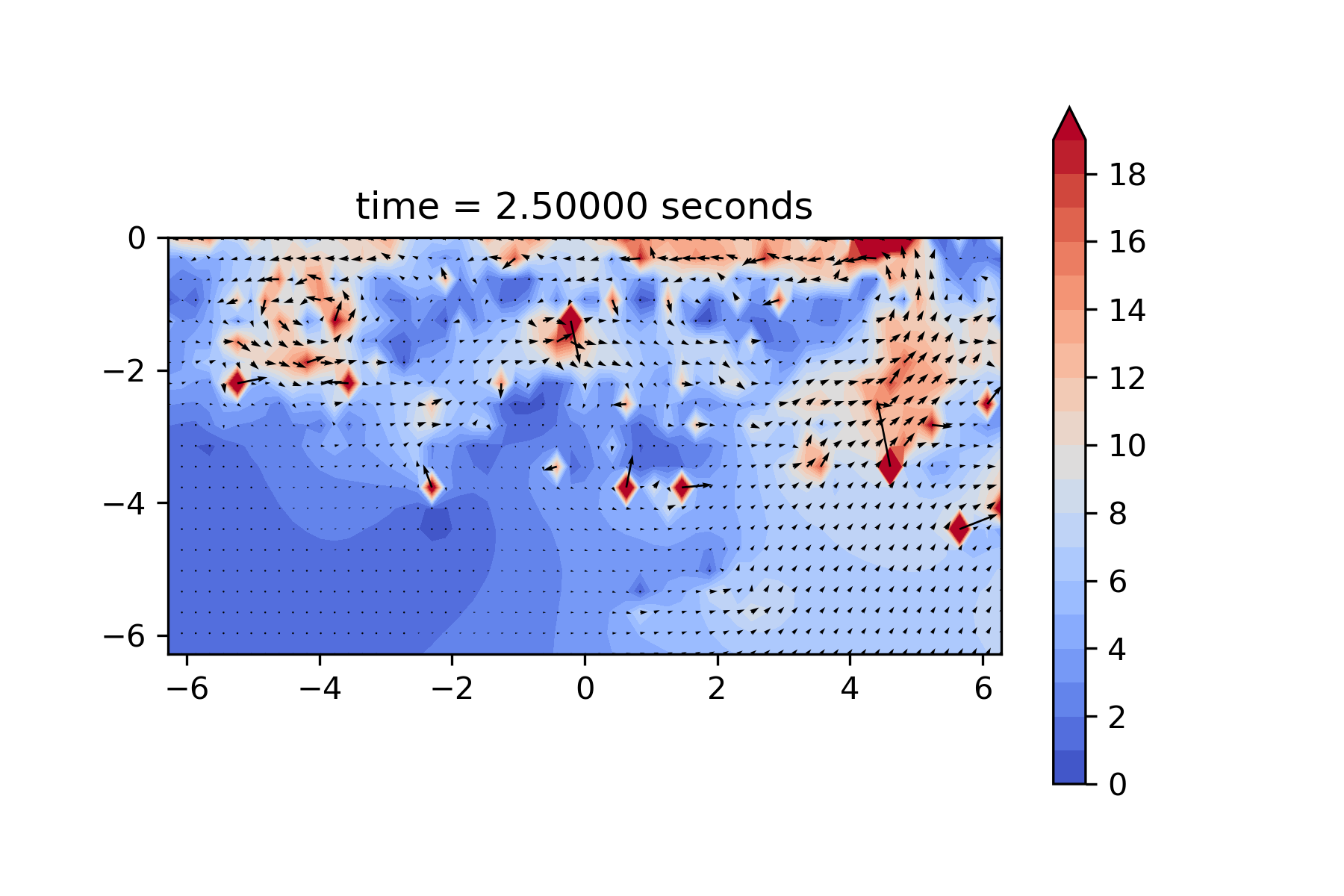}\tabularnewline
\includegraphics[width=4cm,height=3.2cm]{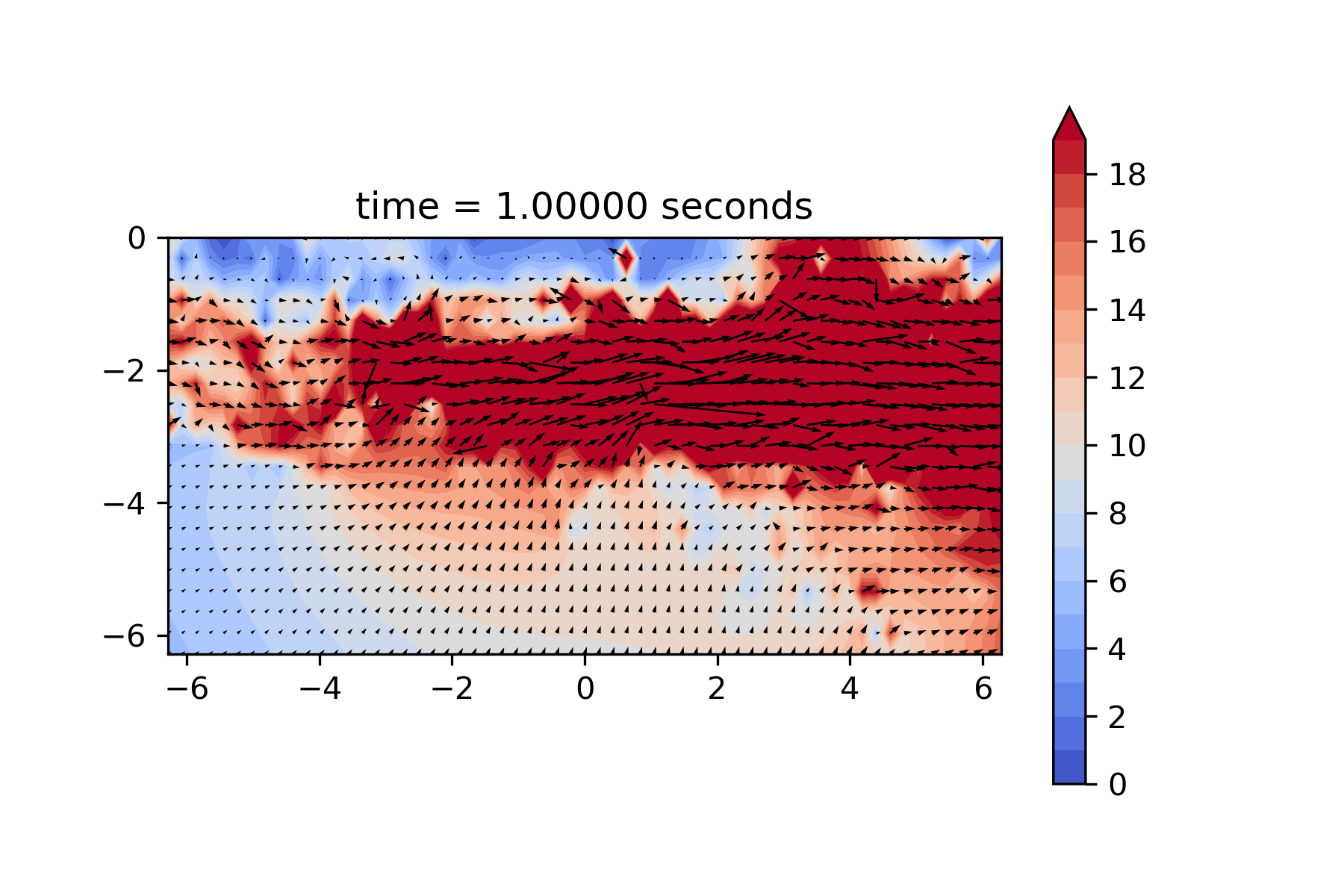}  & \includegraphics[width=4cm,height=3.2cm]{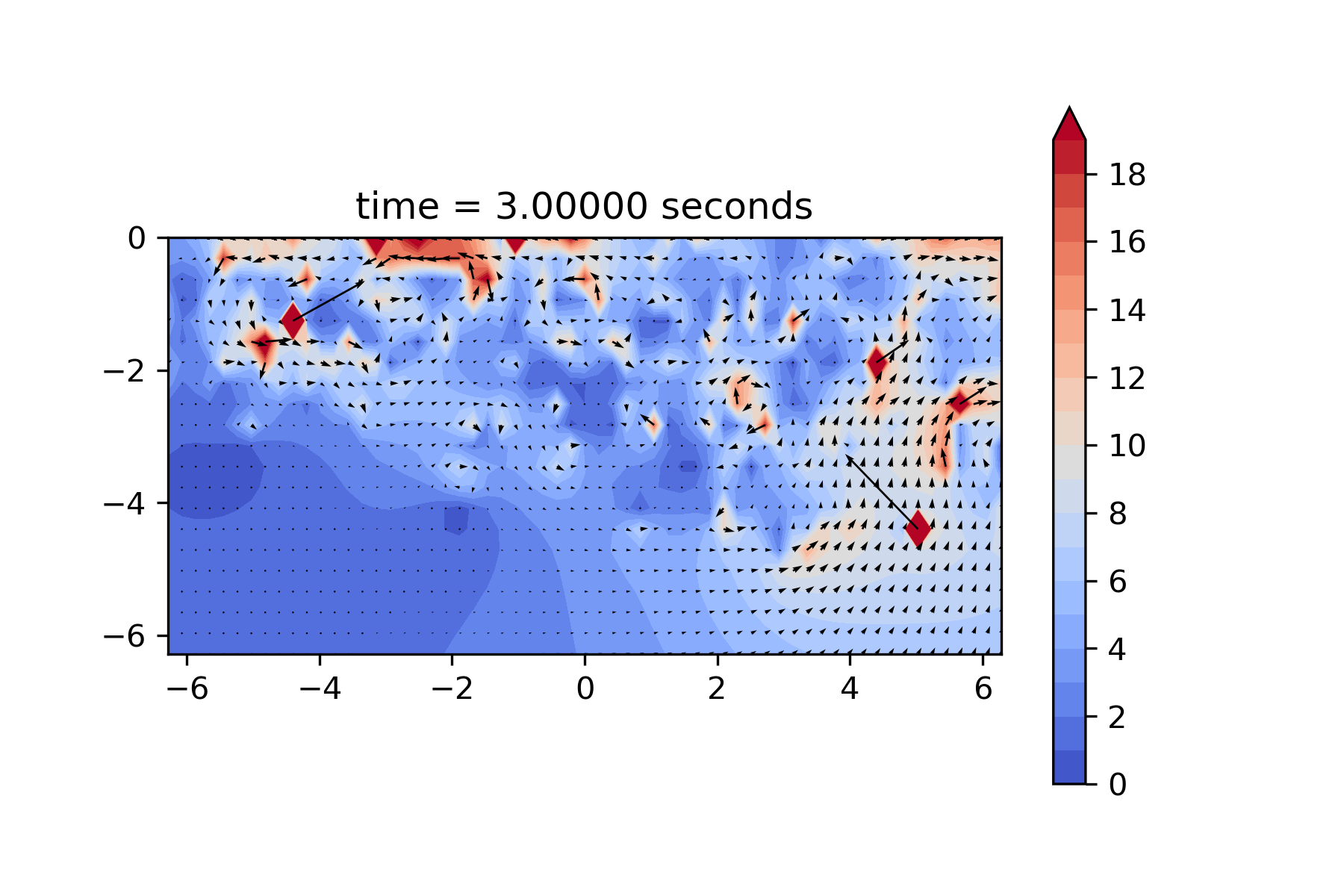}\tabularnewline
\includegraphics[width=4cm,height=3.2cm]{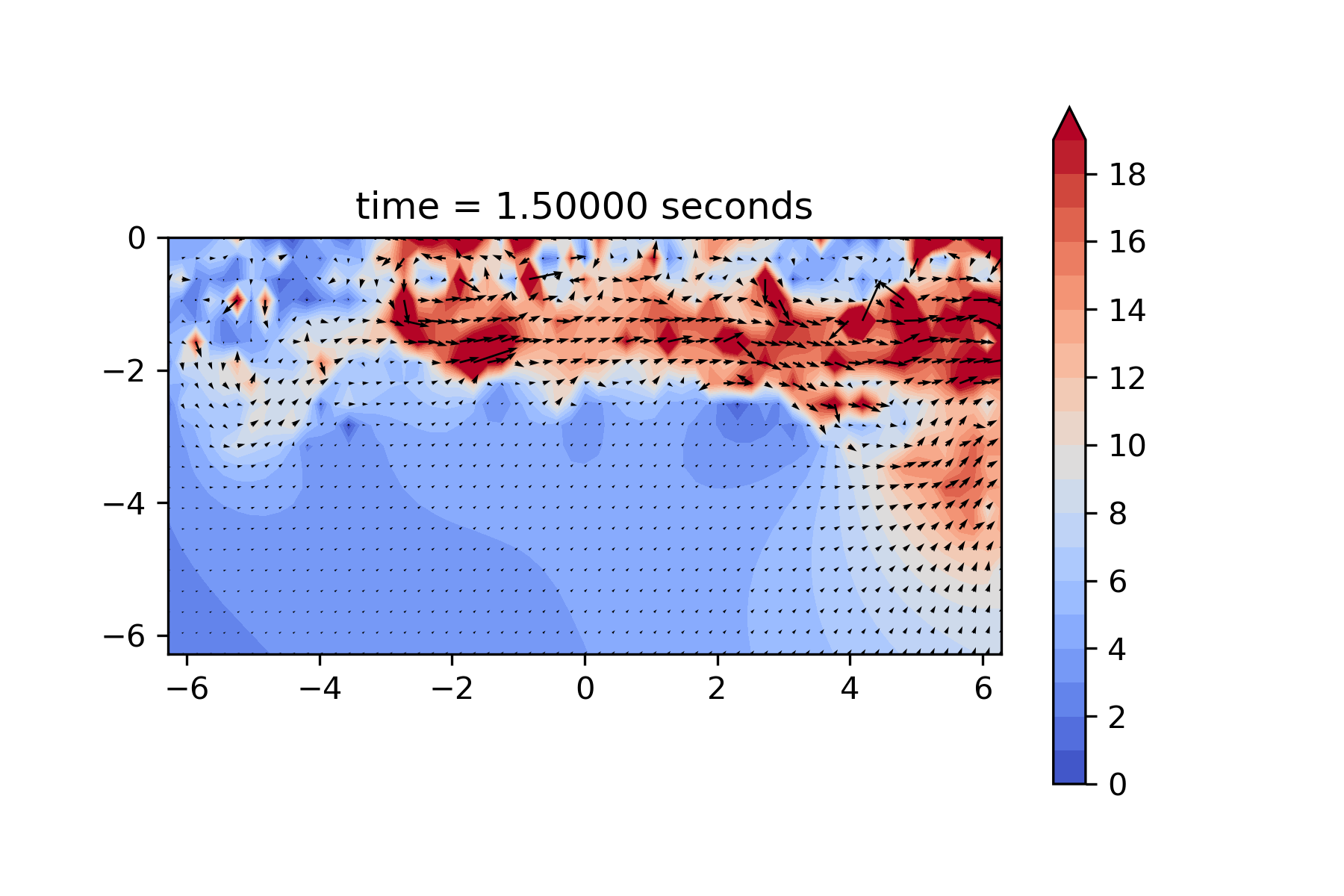}  & \includegraphics[width=4cm,height=3.2cm]{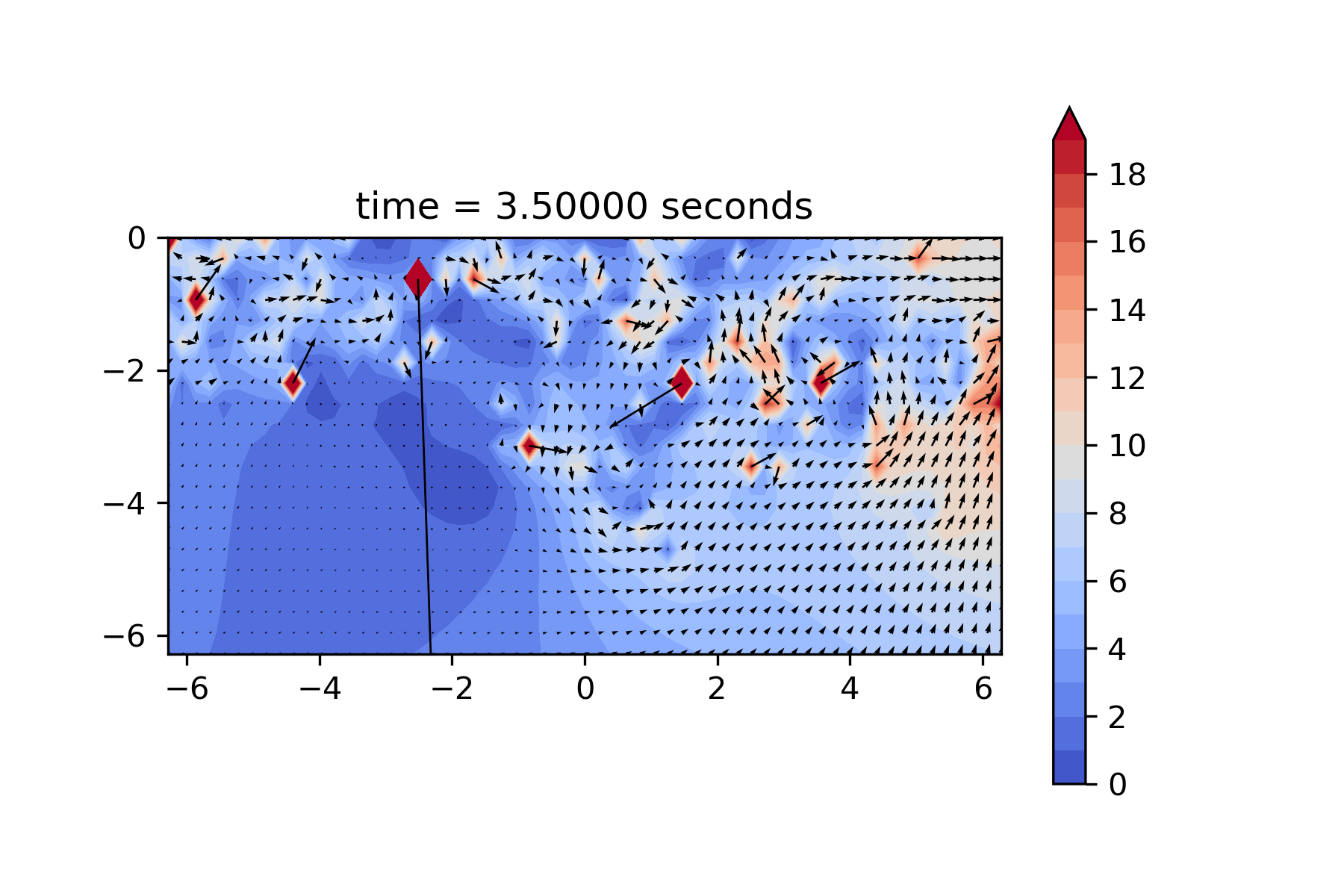}\tabularnewline
\includegraphics[width=4cm,height=3.2cm]{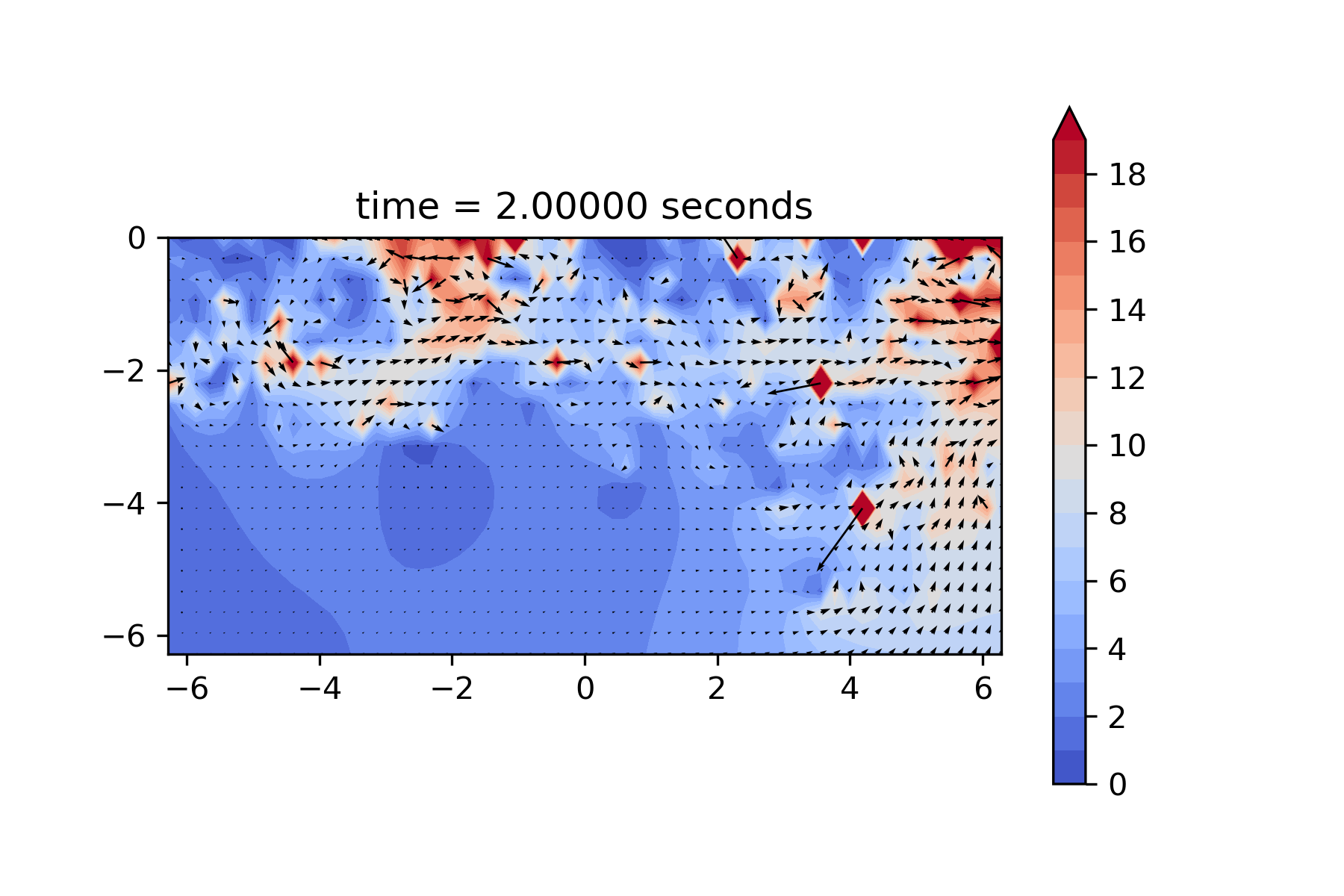}  & \includegraphics[width=4cm,height=3.2cm]{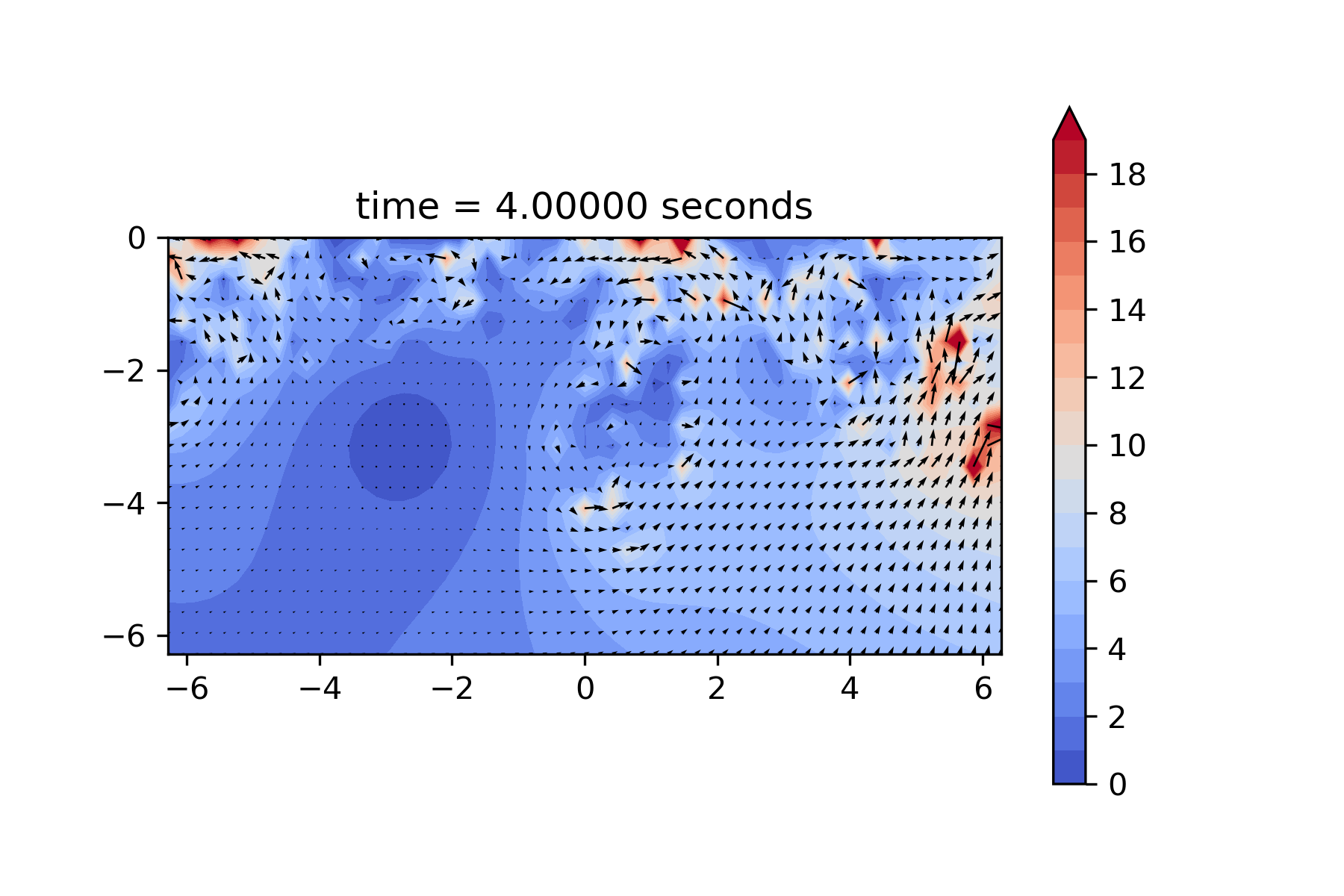}\tabularnewline
\end{tabular}
\end{figure}


\begin{thebibliography}{99}
\bibitem{AndersonGreengard1985} Anderson, C. and Greengard, C. 1985
On vortex methods\emph{. SIAM J. Numer. Anal}. $\mathbf{22}$ (3),
413-440.

\bibitem{BerselliIliescuLayton2006} Berselli, L. C., Iliescu, T.
and Layton, W. J. 2006 \emph{Mathematics of large eddy simulation
of turbulent flows}. Springer.

\bibitem{ChauhanPhilipetl2014} Chauhan, K. Philip, J., De Silva,
C.C.M., Hutchins, N. and Marusic, I. 2014 The turbulent/non-turbulent
interface and entrainment in a boundary layer. \emph{J. Fluid Mech}.
742, 119--151.

\bibitem{Chorin 1973} Chorin, A. J. 1973 Numerical study of slightly
viscous flow. \emph{J. Fluid Mech.} $\mathbf{57}$, 785-796.

\bibitem{CottetKoumoutsakos2000} Cottet, G. -H., and Koumoutsakos,
P. D. 2000 \emph{Vortex methods: theory and practice}. Cambridge University
Press.

\bibitem{Cuvelier-Segal-van Steenhoven1986} Cuvelier, C., Segal,
A. and van Steenhoven, A. A. 1986 \emph{Finite element methods and
Navier-Stokes equations}. D. Reidel Pub. Company.

\bibitem{DawsonMcKeon2019} Dawson, S.T.M. and McKeon, B.J. 2019 On
the shape of resolvent modes in wall-bounded turbulence. \emph{J.
Fluid Mech}. $\mathbf{877}$, 682--716.

\bibitem{Deardorff1974} Deardorff, J. W. 1974 Three-dimensional numerical
study of the height and mean structure of heated planetary boundary
layer. \emph{Boundary-Layer Meteorol}. $\mathbf{7}$, 81-106.

\bibitem{Goodman1987} Goodman, J. 1987 Convergence of the random
vortex method. \emph{Comm. Pure Appl. Math.} $\mathbf{40}$(2), 189-220.

\bibitem{HEISEL2018} Heisel, M., Dasari, T., Liu,Y., Hong, J., Coletti,
F. and Guala, M. 2018 The spatial structure of the logarithmic region
in very-high-Reynolds-number rough wall turbulent boundary layers.
\emph{J. Fluid Mech}. $\boldsymbol{857}$, 704--747.

\bibitem{HeadBandyopadhyay1981} Head, M. R. \& Bandyopadhyay, P.
1981 New aspects of turbulent boundary layer structure. \emph{J. Fluid
Mech}. $\mathbf{107}$, 297--338.

\bibitem{HirschelCousteixKordulla2014} Hirschel, E. H.; Cousteix,
J. and Kordulla, W. 2014 \emph{Three-dimensional attached viscous
flow}. Springer.

\bibitem{Ikeda and Watanabe 1989} Ikeda, N. and S. Watanabe 1989\emph{
Stochastic differential equations and diffusion processes}. Second
Edition. North-Holland Pub. Company.

\bibitem{Keller1978} Keller, H. B. 1978 Numerical methods in boundary-layer
theory. \emph{Ann. Rev. Fluid Mech}. $\mathbf{10}$, 417-33.

\bibitem{LaskariMcKeon} Laskari, A. and McKeon, B. J. 2021 Temporal
characteristics of the probability density function of velocity in
wall-bounded turbulent flows. \emph{J. Fluid Mech.} $\mathbf{913,A6}$,
doi:10.1017/jfm.2020.1163.

\bibitem{LesieurMetaisComte2005} Lesieur, M., M�tais, O. and Comte,
P. 2005 \emph{Large-Eddy Simulations of Turbulence}. Springer.

\bibitem{Lilly1967} Lilly, D. K. 1967 The representation of small-scale
turbulence in numerical simulation experiments. In H. H. Goldstine
(Ed.), \emph{Proc. IBM Scientific Computing Symp. on Environmental
Sciences}, pp. 195-210. Yorktown Heights, NY: IBM.

\bibitem{Long1988} Long, D. G. 1988 Convergence of the random vortex
method in two dimensions. \emph{J. of Amer. Math. Soc}. $\mathbf{1}$(4
), 779-804.

\bibitem{Majda and Bertozzi 2002} Majda, A. J. and Bertozzi A. L.
2002\emph{ Vorticity and incompressible flow}. Cambridge University
Press.

\bibitem{MCKEON2010} McKeon, B.J. and Sharma, A.S. 2010 A critical-layer
framework for turbulent pipe flow. \emph{J. Fluid Mech.} $\mathbf{658}$,
336--382.

\bibitem{MoinMashesh1998} Moin, P. and Mahesh, K. 1998 Direct numerical
simulation: a tool in turbulence research. \emph{Annu. Rev. Fluid
Mech.} $\mathbf{30}$, 539--78

\bibitem{OrszagPatterson1972} Orszag S.A. and Patterson G.S. 1972
Numerical simulation of three-dimensional homogeneous isotropic turbulence.
\emph{Phys. Rev. Lett}. $\mathbf{28}$, 76--79.

\bibitem{Pope2000} Pope, S. B. 2000 \emph{Turbulent flows}. Cambridge
University Press.

\bibitem{Prandtl1904} Prandtl, L. 1904 �ber Fl�ssigkeitsbewegung
bei sehr kleiner Reibung. \emph{Proc. Third Intern. Math. Congress},
Heidelberg, 848-491.

\bibitem{RaiMoin1993}Rai, M. M. \& Moin, P. 1993 Direct numerical
simulation of transition and turbulence in a spatially evolving boundary
layer.\emph{ J. Comput. Phys}. $\mathbf{109}$, 169--192.

\bibitem{Schlichting9th-2017} Schlichting, H. and Gersten, K. 2017
\emph{Boundary-layer theory} (Ninth Edition). Springer.

\bibitem{Sengupta-Bhaumik2019} Sengupta, T. K. and Bhaumik, S. 2019
\emph{DNS of wall-bounded turbulent flows -- a first principle approach}.
Springer.

\bibitem{Smagorinsky1963} Smagorinsky, J. 1963 General circulation
experiments with the primitive equations: I. the basic equations.
\emph{Mon. Weather Rev.} $\mathbf{91}$, 99-164.

\bibitem{Spalart1988} Spalart, P. R. 1988 Direct simulation of a
turbulent boundary layer up to $R_{\theta}=1410$. \emph{J. Fluid
Mech}. $\mathbf{187}$, 61--98.

\bibitem{SpalartWatmuff1993} Spalart, P. R. and Watmuff, J. H. 1993
Experimental and numerical study of a turbulent boundary layer with
pressure gradients. \emph{J. Fluid Mech}. $\mathbf{249}$, 337--371.

\bibitem{Temam2000} Temam, R. 1977\emph{ Navier-Stokes equations}:
\emph{theory and numerical analysis}. Amer. Math. Soc., Providence,
RI. North-Holland, Amsterdam.

\bibitem{Taylor1921} Taylor, G. I. 1921 Diffusion by continuous movements.
\emph{Proc. Lond. Math. Soc}. $\mathbf{20}$, 196.

\bibitem{Wesseling2001} Wesseling, P. 2001 \emph{Principles of computational
fluid dynamics}. Springer-Verlag.

\bibitem{WuMoin2008}Wu, X. and Moin, P. 2008 A direct numerical simulation
study on the mean velocity characteristics in turbulent pipe flow.
\emph{J. Fluid Mech}. $\mathbf{608}$, 81--112.

\bibitem{WuMoin2009}Wu, X. and Moin, P. 2009 Direct numerical simulation
of turbulence in a nominally zero-pressure-gradient flat-plate boundary
layer. \emph{J. Fluid Mech}. $\mathbf{630}$, pp. 5--41.

\bibitem{WuMoinHickey2014} Wu, X., Moin, P. and Hickey, J. P. 2014
Boundary layer bypass transition. \emph{Physics of Fluids}, $\mathbf{26}$,
091104. 
\end{thebibliography}
\end{document}